\documentclass[12pt]{article}

\usepackage{amsmath,amssymb,amsfonts,amsthm,bbm}
\usepackage{graphicx}
\usepackage{float}

\pdfoutput=1

\usepackage[utf8]{inputenc} 
\usepackage[T1]{fontenc}    
\usepackage{hyperref}       
\usepackage{url}            
\usepackage{booktabs}       
\usepackage{amsfonts}       
\usepackage{nicefrac}       
\usepackage{microtype}      

\usepackage{natbib}

\makeatletter
\def\singlespace{\def\baselinestretch{1}\@normalsize}

\makeatletter
\def\singlespace{\def\baselinestretch{1}\@normalsize}

\title{Taming heavy-tailed features by shrinkage}

\newcommand{\bfm}[1]{\ensuremath{\mathbf{#1}}}
     \def\bA{\bfm A}     \def\cA{{\cal  A}}     
          \def\cB{{\cal  B}}     
          \def\cC{{\cal  C}}

     \def\bH{\bfm H}          
          \def\cI{{\cal  I}}

          \def\cL{{\cal  L}}

          \def\cS{{\cal  S}}

\def\bv{\bfm v}               
               
\def\bx{\bfm x}     \def\bX{\bfm X}          
\def\by{\bfm y}               
     \def\bZ{\bfm Z}          
\def\bzero{\bfm 0}

\newcommand{\bfsym}[1]{\ensuremath{\boldsymbol{#1}}}
       \def \bbeta    {\bfsym{\beta}}

       \def \bDelta   {\bfsym{\Delta}}

\def \bSigma   {\bfsym{\Sigma}}

\renewcommand{\hat}{\widehat}

\def \heps     {\hat{\heps}}

\DeclareMathOperator*{\argmin}{argmin}

\DeclareMathOperator{\E}{E}

\DeclareMathOperator{\var}{var}

\def \var   {\mbox{var}}

 at 8truept

\def\today{\ifcase\month\or
  January\or February\or March\or April\or May\or June\or
  July\or August\or September\or October\or November\or December\fi
  \space\number\day, \number\year}

\def \newpage {\vfill\eject}

\newdimen\biblioindent\biblioindent=30pt
\newcommand{\beq}  {\begin{equation}}
\newcommand{\eeq}  {\end{equation}}
\newcommand{\beqn} {\begin{eqnarray}}
\newcommand{\eeqn} {\end{eqnarray}}
\newcommand{\beqnn}{\begin{eqnarray*}}
\newcommand{\eeqnn}{\end{eqnarray*}}

\newtheorem{lem}{Lemma}

\newtheorem{cor}{Corollary}

\newtheorem{thm}{Theorem}
\newtheorem{rem}{Remark}
\newcounter{CondCounter}

\newcommand{\lonenorm}[1]{\lVert#1\rVert_1}
\newcommand{\ltwonorm}[1]{\lVert#1\rVert_2}
\newcommand{\lfournorm}[1]{\lVert#1\rVert_4}
\newcommand{\opnorm}[1]{\lVert#1\rVert_{\mathrm{op}}}

\newcommand{\supnorm}[1]{ \lVert#1  \rVert_{\max}}
\newcommand{\inn}[1]{\langle #1 \rangle}

\def \RR	{\mathbb{R}}

\def \PP      {\mathbb{P}}
\def \EE      {\mathbb{E}}

\def \ind {1}

\author{%
  Ziwei Zhu, Wenjing Zhou\\
  \normalsize
  Department of Statistics, University of Michigan, Ann Arbor\\
}

\begin{document}

\maketitle

\begin{abstract}

In this work, we focus on a variant of the generalized linear model (GLM) called corrupted GLM (CGLM) with heavy-tailed features and responses. To robustify the statistical inference on this model, we propose to apply $\ell_4$-norm shrinkage to the feature vectors in the low-dimensional regime and apply elementwise shrinkage to them in the high-dimensional regime. Under bounded fourth moment assumptions, we show that the maximum likelihood estimator (MLE) based on the shrunk data enjoys nearly the minimax optimal rate with an exponential deviation bound. Our simulations demonstrate that the proposed feature shrinkage significantly enhances the statistical performance in linear regression and logistic regression on heavy-tailed data. Finally, we apply our shrinkage principle to guard against mislabeling and image noise in the human-written digit recognition problem. We add an $\ell_4$-norm shrinkage layer to the original neural net and reduce the testing misclassification rate by more than $30\%$ relatively in the presence of mislabeling and image noise. 

\end{abstract}

\section{Introduction}

Heavy-tailed data abound in modern data analytics. For instance, financial log-returns and macroeconomic variables usually exhibit heavy tails (\citet{Con01}). In a genomic study, microarray data are always wildly fluctuated (\citet{LDG03}, \citet{PHo05}). In deep learning, features learned by deep neural nets are generated via highly nonlinear transformation of the original data and thus have no guarantee of exponential-tailed distribution. These real-world cases contradict the common sub-Gaussian or sub-exponential conditions in the statistics literature.  A series of questions thus arise: with heavy-tailed data, can we still achieve good statistical properties of the previous standard estimators or testing statistics? If not, is there a solution to overcome heavy-tailed corruption and achieve equally well statistical performance as with exponential-tailed data?
		
To answer these questions, perhaps the easiest statistical problem to start with is the mean estimation problem. It turns out surprisingly, as first pointed out by \citet{Cat12}, that the empirical mean is far from optimal when data have only a few finite moments. \citet{Cat12} proposed a novel M-estimator for the population mean and revealed its sub-Gaussian behavior around the true mean under merely bounded second moment assumptions. The correspondent score function is constructed to be logarithmic with respect to the deviation when it is large, thereby being insensitive to outliers and yielding a robust M-estimator. Later on \citet{Min15}, \citet{DLL16} and \citet{HSa16} established a similar sub-Gaussian concentration property for the median-of-means estimator (\citet{NYD82}). Particularly, \citet{Min15} and \citet{HSa16} consider the median-of-means approach under general metric spaces. 
		
\begin{figure*}[t]
	\centering
	\begin{tabular}{ccc}
 \hspace{-.8cm}\includegraphics[scale=.30]{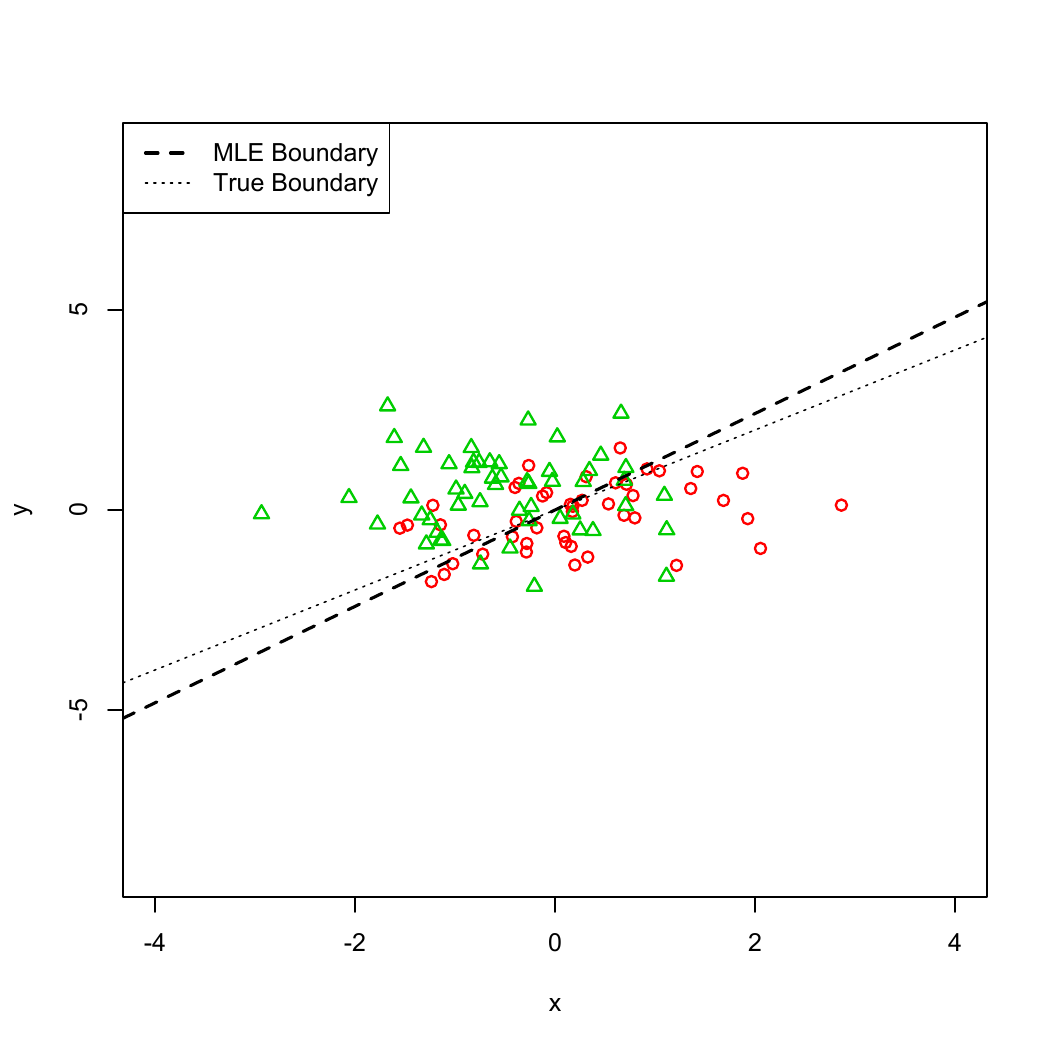} & 
 \hspace{-1.1cm} \includegraphics[scale=.30]{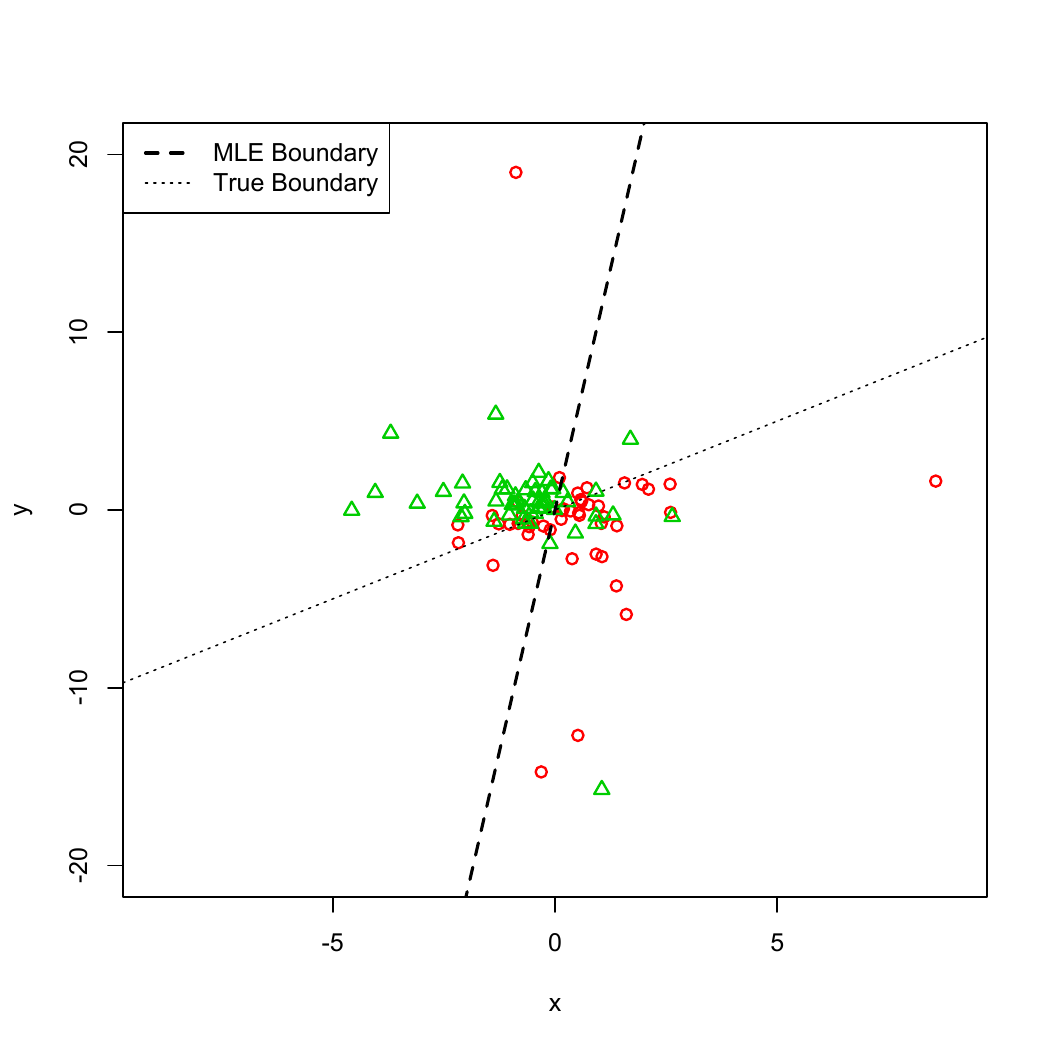} & 
 \hspace{-1.1cm} \includegraphics[scale=.30]{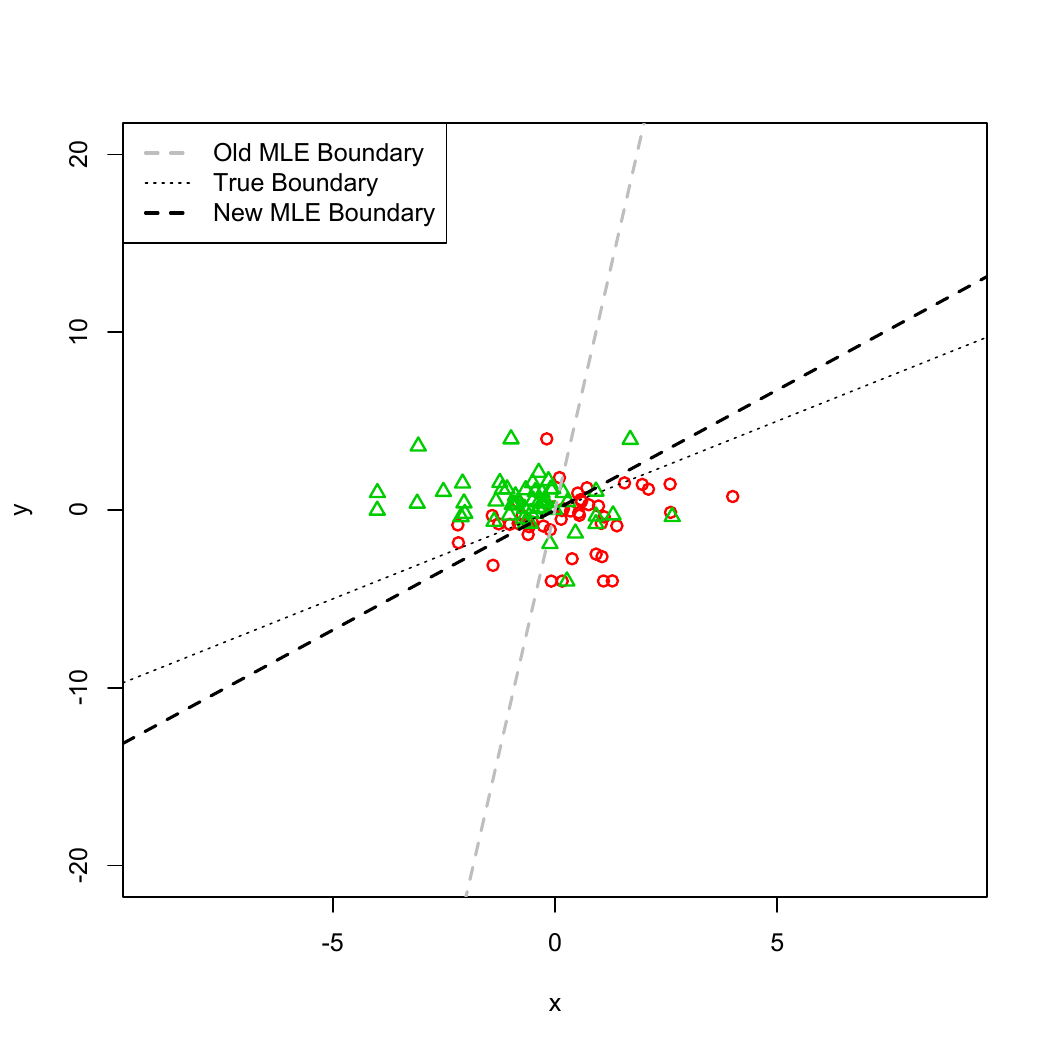} \\
 \hspace{-.5cm} (a) Gaussian features & 
 \hspace{-.7cm} (b) Student's $t_2$ features  & 
 \hspace{-.8cm} (c) Shrunk Student's $t_2$ features
	\end{tabular}
	\caption{Logistic Regression with $10\%$ mislabeled data based on different features}
	\label{fig:1}
	\vspace{-.5cm}
\end{figure*}
		
Beyond the mean estimation problem, robust risk minimization and the median-of-means approach are proved to be successful under a great variety of problem setups with heavy-tailed data, e.g., covariance matrix or general matrix estimation (\citet{Min18, FWZ17}), empirical risk minimization (\citet{BJL15, HSa16}), low-dimensional regression and high-dimensional sparse linear regression (\citet{FLW16, PSZ19, SZF17, WZZ18}), low-rank matrix recovery (\citet{FWZ17}) and so forth. 
		
Despite heated research on statistics with heavy-tailed data, very few have studied the effect of heavy tails of features or designs in regression. It remains unclear whether widely spread features or designs are blessings or curses to statistical efficiency. This motivates us to consider a variant of the generalized linear model (GLM) called corrupted GLM (CGLM) that accommodates both heavy-tailed designs and responses. The CGLM allows extra random corruption on the response of the traditional GLM, thereby enjoying much broader model capacity and embraces a myriad of important real-world problems. 
		
One key message of our paper is that heavy-tailed features can aggravate the corruption on the response and jeopardize standard statistical approaches. To further illustrate this point, Panels (a) and (b) of Figure \ref{fig:1} contrast the performance of the standard MLE on light-tailed features and heavy-tailed features under a logistic regression model. When the data points are widely spread as in Panel (b), the boundary derived from the MLE deviates far from the true boundary. When the data points are Gaussian, however, Panel (a) shows nearly perfect alignment between the MLE boundary and the true boundary. The reason for this difference is that the outliers, especially those mislabeled, have severe influence on the log-likelihood and can easily destroy the validity of the MLE.
		
 To tame the heavy-tails of the features, we propose to shrink the features before calculating the M-estimator. Given feature vectors $\{\bx_i \in \RR^d\}_{i=1}^n$, a threshold value $\tau$ and a norm $\|\cdot\|$ on the feature space, the shrunk features $\{\widetilde\bx^s_i\}_{i=1}^n$ are defined as: 
		 \[
		 	\widetilde \bx^s_i = \min(\|\bx_i\|, \tau) \frac{\bx_i}{\|\bx_i\|}.
		 \]
In short, we restrict $\|\widetilde\bx^s_i\|$ below the level $\tau$. 
In the sequel, we illustrate both theoretically and numerically that the feature shrinkage trades little bias for great variance reduction such that the resulting MLE achieves (nearly) the minimax optimal statistical rate. Panels (b) and (c) of Figure \ref{fig:1} compare the performance of MLE based on original heavy-tailed features and shrunk features. One can see that after feature shrinkage, the new MLE boundary becomes much more aligned with the true boundary than the original one, because the shrinkage mitigates the perturbation of the outliers on the log-likelihood. Note that similar ideas have been explored to overcome adversarial corruption on features. For example, \citet{CCM13} used the trimmed inner product to robustify standard high-dimensional regression methods and established strong statistical guarantees while allowing a certain fraction of observations to be arbitrarily corrupted. \citet{FXM14} proposed to ignore observations with large feature values to prevent adversarial feature corruption in logistic regression and binary classification problems. The major difference between our work and theirs is that our focus is tail behavior, rather than corruption, of features in regression problems. We assume that the features have only few bounded moments, while \citet{CCM13} and \citet{FXM14} assume them to be sub-Gaussian. Our theory does not assume any corruption on the features; all the corruption in this paper is imposed on responses.
		
The rest of the paper is organized as follows. In Section 2, we elucidate the CGLM and the log-likelihood based on the shrunk data. In Section 3, we introduce specific feature shrinkage methods for different regimes and present our main theoretical results. Under the low-dimensional regime, we prove that the MLE based on $\ell_4$-norm shrunk features enjoys the same optimal statistical rate as the standard MLE with sub-Gaussian features up to a $(\log n) ^ {1 / 2}$ factor. For high-dimensional models, we show that the $\ell_1$-regularized MLE based on elementwise shrunk features achieve exactly the minimax optimal rate. One technical contribution worth emphasis is that we provide a rigorous justification of the (restricted) strong convexity of the negative likelihood based on shrunk features. In Section 4, we demonstrate the numerical superiority of our proposed estimators over the standard MLEs under both low-dimensional and high-dimensional regimes. We investigate two important problem setups: linear regression with heavy-tailed noise and binary logistic regression with mislabeled data. Finally, motivated by the shrinakge principle, we add an $\ell_4$-norm shrinkage layer to a convolutional neural network to classify human-written digits in the MNIST dataset. We show the significant improvement of the new architecture in the presence of mislabeling and image noise. 

\section{Problem setup}

In this section, we formulate the corrupted GLM as aforementioned. Recall the definition of the standard GLM with the canonical link. Suppose we have $n$ observations $\{(y_i, \bx_i)\}_{i=1}^n$, where $y_i$ is the response and $\bx_i$ is the feature vector valued in $\RR ^ d$. Under the GLM with the canonical link, the probability density function of the response $y_i$ is defined as 
	\begin{equation}
		\label{eq:2.1}
		\begin{aligned}
		f_n(\by; \bX, {\bbeta^*})& =\prod\limits_{i=1}^n f(y_i; \eta^*_i) = \prod\limits_{i=1}^n\biggl\{c(y_i)\exp\biggl(\frac{y_i\eta^*_i-b(\eta^*_i)}{\phi}\biggr)\biggr\},
		\end{aligned}
	\end{equation}
	
	where $\by= (y_1, \cdots, y_n)^\top$, $\bX=(\bx_1, \cdots, \bx_n)^\top$, $\bbeta^*\in \RR^d$ is the  regression coefficient vector, $\eta_i^*:=\bx_i^\top\bbeta^*$, $b(\cdot)$ is a known function that is twice differentiable with a positive second derivative and $\phi>0$ is the dispersion parameter. The negative log-likelihood corresponding to (\ref{eq:2.1}) is given, up to an affine transformation, by
	\begin{equation}
		\begin{aligned}
		\label{eq:2.2} 
		\ell_n(\bbeta)=\frac{1}{n}\sum\limits_{i=1}^n -y_i\bx_i^\top\bbeta+b(\bx_i^\top\bbeta) =\frac{1}{n}\sum\limits_{i=1}^n -y_i\eta_i+b(\eta_i)=\frac{1}{n}\sum\limits_{i=1}^n \ell_i(\bbeta),
		\end{aligned}
	\end{equation}
	and the gradient and Hessian of $\ell_n(\bbeta)$ are respectively
	\beq
		\label{eq:2.3}
			\nabla\ell_n(\bbeta)=-\frac{1}{n}\sum\limits_{i=1}^n (y_i - b'(\bx_i^\top\bbeta^*))\bx_i~~~\text{and}~~~\nabla^{2}\ell_n(\bbeta)=\frac{1}{n} \sum\limits_{i=1}^n b''(\bx_i^\top\bbeta^*)\bx_i\bx_i^\top. 
	\eeq
	Note that $b'(\bx_i^\top\bbeta ^ *) = \EE(y_i | \bx_i)$. For ease of notation, we write the empirical hessian $\nabla^2 \ell_n(\bbeta)$ as $\bH_n(\bbeta)$ and $\EE (b''(\bx_i^\top\bbeta)\bx_i\bx_i^\top)$ as $\bH(\bbeta)$. 
	
	Under a CGLM, for the $i$th observation we can only observe its corrupted response 
	\beq	
		\label{eq:cglm}
		z_i = y_i + \epsilon_i
	\eeq
	rather than the original response $y_i$, where $\epsilon_i$ is random noise. We emphasize that introducing  $\epsilon_i$ significantly improves the flexibility of the original GLM, such that now the response is not limited within the exponential family. The CGLM embraces many more real-world problems with complex structures, e.g., the linear regression model with heavy-tailed noise, the logistic regression with mislabeled samples and so forth.
	
	To handle the heavy-tailed features and noise on the response, we propose to shrink the data $\{(z_i, \bx_i)\}_{i=1}^n$ first and use them to construct the log-likelihood \eqref{eq:2.2}. Formally, define
	\beq
		\label{eq:2.4}
		\widetilde \ell_n(\bbeta) := \frac{1}{n}\sum\limits_{i=1}^n -\widetilde z_i\widetilde \bx_i^\top\bbeta+b(\widetilde\bx_i^\top\bbeta). 	\eeq
	 We denote the hessian matrix of $\widetilde\ell_n(\bbeta)$ by $\widetilde\bH_n(\bbeta)$ and its population version $\EE \widetilde\bH_n(\bbeta)$ by $\widetilde\bH(\bbeta)$. In the next section, we elucidate the specific shrinkage methods to construct $\widetilde \bx_i$ and $\widetilde z_i$ in both low-dimensional and high-dimensional regimes and explicitly derive the statistical error rates of the MLE based on $\widetilde\ell_n(\bbeta)$. 
	
\section{Main results}

\subsection{Notation}
	
	Here we collect all the notation that we use in the sequel. We use regular letters for scalars, bold regular letters for vectors and bold capital letters for matrices. Denote the $d$-dimensional Euclidean unit sphere by $\cS^{d-1}$. Denote the Euclidean and $\ell_1$-norm balls with the center $\bbeta^*$ and radius $r$ by $\cB_2(\bbeta^*, r)$ and $\cB_1(\bbeta^*, r)$ respectively. We write the set $\{1, \cdots, d\}$ as $[d]$. For two scalar sequences $\{a_n\}_{n \ge 1}$ and $\{b_n\}_{n \ge 1}$, we say $a_n \asymp b_n$ if there exist two universal constants $C_1$ and $C_2$ such that $C_1 b_n \le a_n  \le C_2 b_n$ for all $n \ge 1$. We use $\ltwonorm{\bv}$, $\lonenorm{\bv}$ and $\lfournorm{\bv}$ to denote the Euclidean norm, $\ell_1$-norm and $\ell_4$-norm of $\bv$ respectively. Particularly, recall that $\lfournorm{\bx_i} := (\sum_{j=1}^d x^4_{ij})^{1/4}$. For a matrix $\bA$, we use $\opnorm{\bA}$ and $\supnorm{\bA}$ to denote the operator norm and elementwise max-norm of $\bA$ respectively and use $\lambda_{\min}(\bA)$ to denote the minimum eigenvalue of $\bA$. For any $\bbeta^* \in \RR^d$ and any differential map $f: \RR^d \rightarrow \RR$, define the first-order Taylor remainder of $f(\bbeta)$ at $\bbeta = \bbeta^*$ to be
	 \[
	 	\delta f (\bbeta ; \bbeta^*) := f(\bbeta) - f(\bbeta^*) - \nabla f(\bbeta^*)^\top (\bbeta- \bbeta^*). 
	 \]
	For a set of random variables $\{X_i\}_{i \in \cI}$, we say that they are i.i.d. if they are independent and identically distributed. We refer to some quantities as \emph{constants} if they are independent of the sample size $n$, the dimension $d$ and the sparsity $s$ of $\bbeta^*$ in the high-dimensional regime. 
	
	\subsection{Low-dimensional regime}
	
	The standard MLE estimator is defined as $\widehat\bbeta := \argmin_{\bbeta \in \RR^d} \ell_n(\bbeta)$, where $\ell_n(\cdot)$ is characterized as in \eqref{eq:2.2}. It is well established that under a standard GLM with bounded features, $\widehat\bbeta$ enjoys $(d/n) ^ {1 / 2}$-consistency to the true parameter $\bbeta^*$ in terms of the Euclidean norm. However, when the feature vectors have only bounded moments, there is no guarantee of $(d/n) ^ {1 / 2}$-consistency any more, let alone further perturbation on the response. To overcome the disruption due to heavy-tailed data, we apply $\ell_4$-norm shrinkage to the feature vectors. Construct
	\beq
		\label{eq:l4_norm_shrinkage}
		\widetilde \bx_i := \frac{\min(\lfournorm{\bx_i}, \tau_1)}{\lfournorm{\bx_i}} \bx_i 
	\eeq
 	and
	\beq
		\label{eq:clip_response}
		\widetilde z_i := \min(|z_i|, \tau_2) z_i/|z_i|, 
	\eeq
	where $\tau_1$ and $\tau_2$ are predetermined thresholds. Clipping on the response is natural; when $|z_i|$ is abnormally large, clipping reduces its magnitude to prevent corruption by $\epsilon_i$. Here we explain more on why we shrink features in terms of the $\ell_4$-norm rather than other norms. The $\ell_4$-norm shrinkage has been proven to be successful in low-dimensional covariance estimation in \citet{FWZ17}. Theorem 6 therein shows that when data have only bounded fourth moments, the $\ell_4$-norm shrinkage sample covariance enjoys an operator-norm rate of order $O_{\PP}\{(d \log d/n) ^ {1 / 2}\}$ in estimating the population covariance matrix. This inspires us to apply similar $\ell_4$-norm shrinkage to heavy-tailed features to ensure that the empirical hessian $\widetilde \bH_n(\bbeta)$ is close to its population version $\bH(\bbeta)$ and thus well-behaved. After data shrinkage and clipping, we minimize the negative log-likelihood based on the new data $\{\widetilde z_i, \widetilde \bx_i\}_{i=1}^n$ to derive the M-estimator, i.e., we choose $\widetilde\bbeta := \argmin_{\bbeta \in \RR^d}  \widetilde\ell_n(\bbeta)$ to estimate $\bbeta^*$, where $\widetilde\ell_n(\bbeta)$ is defined as in \eqref{eq:2.4}.
 	
 	We first establish the uniform strong convexity of $\widetilde\ell_n(\bbeta)$ over $\bbeta\in \cB_2(\bbeta^*, r)$ (up to some small tolerance term) that is crucial to our subsequent statistical analysis. 
	\begin{lem}
		\label{lem:1}
		Suppose the following conditions hold: (1) $\forall i \in [n]$, $b''(\bx_i^\top \bbeta^*) \le M< \infty$, and $\forall\ \omega >0,\ \exists \ m(\omega)>0$ such that $b''(\eta) \ge m(\omega) > 0$ for $|\eta| \le \omega$;  (2) $\EE \bx_i = \bzero$, $\lambda_{\min}(\EE \bx_i\bx_i^\top) \ge \kappa_0 > 0$ and $\EE (\bv^\top\bx_i)^4 \le R < \infty$ for all $\bv\in \cS^{d-1}$; (3) $\ltwonorm{\bbeta^*} \le L < \infty$. Choose the shrinkage threshold $\tau_1 \asymp (n/\log n)^{1/4}$. For any $0 < r < 1$ and $t>0$, when $(d\log n / n) ^ {1 / 2}$ is sufficiently small, we have with probability at least $1 - 2\exp(-t)$ that for all $\bDelta \in \RR^d$ such that $\ltwonorm{\bDelta} \le r$, 
		\[
			\begin{aligned}
			 \delta \widetilde\ell_n(\bbeta^* + \bDelta; \bbeta^*) \ge \kappa \ltwonorm{\bDelta}^2 - C r^2\biggl \{ \biggl ( \frac{t}{n}\biggr ) ^ {1 / 2} + \biggl ( \frac{d}{n} \biggr ) ^ {1 / 2}\biggr \}, 
			\end{aligned}
		\] 
		where $\kappa$ and $C$ are constants. 
	\end{lem}
	
	\begin{rem}
		Here we explain the conditions of Lemma \ref{lem:1}. Condition (1) assumes that the response from the GLM has bounded variance and is non-degenerate when $\eta$ is bounded. Note here that we {do not assume a uniform lower bound} of $b''(\eta)$. $m(\omega)$ is allowed to decay to zero as $\omega \rightarrow \infty$. Condition (2) says that the population covariance matrix of the design vector $\bx_i$ is positive definite and $\bx_i$ has bounded fourth moment. Condition (3) is natural: it holds if we have $\var(\bx_i^{\top}\bbeta^*) < \infty$ and $\lambda_{\min}(\E \bx_i\bx_i^{\top}) \ge \kappa_0 > 0$.  Note that the ordinary least square (OLS) estimator has been shown to enjoy consistency under similar bounded fourth moment conditions (\citet{HKZ12}, \citet{ACa11}, \citet{Oli16}). Theorem 1 later establishes a similar result for the CGLM. 
	\end{rem}
	\begin{rem}
		In the proof of Theorem \ref{thm:1}, we let the radius of the local neighborhood $r$ here decay to zero so that the tolerance term $r^2\{(t / n) ^ {1 / 2} + (d / n) ^ {1 / 2}\}$ is negligible. 
	\end{rem}
    
    We are now in position to present the statistical rate of $\ltwonorm{\widetilde\bbeta - \bbeta^*}$. 
	
	\begin{thm}
		\label{thm:1}
		Suppose the conditions of Lemma \ref{lem:1} hold. We further assume that (1)  $\EE z_i^4 \le M_1< \infty$; (2) $\ltwonorm{\EE [\epsilon_i\bx_i]} \le M_2 (d / n) ^{1 / 2} $ for some constant $M_2$. Choose $\tau_1, \tau_2 \asymp (n/\log n)^{1/4}$. There exists a constant $C > 0$  such that for any $\xi > 1$, 
		\[
			\PP\biggl \{ \ltwonorm{\widetilde\bbeta - \bbeta^*} \ge  C\xi \biggl (  \frac{d \log n}{n}\biggr ) ^ {1 / 2}\biggr \} \le 3n^{1 - \xi} . 
		\]
	\end{thm}
	
	\begin{rem}
		Condition 1 requires merely bounded fourth moments of the response from CGLM. Condition 2 requires the additional corruption to be nearly uncorrelated with the design, which is trivially satisfied if $\E(\epsilon_i | \bx_i) = 0$. 
	\end{rem}

    In some cases, the covariance between $\epsilon_i$ and $\bx_i$ does not vanish as $n$ and $d$ grow. For example, in binary logistic regression with mislabeling, we have that 
	\beq
		\label{eq:3.2}
		\begin{aligned}
			\PP(\epsilon_i = -1 | y_i = 1) = p, \PP(\epsilon_i = 0 | y_i = 1) = 1-p,\\
			\PP(\epsilon_i = 1 | y_i = 0) = p, \PP(\epsilon_i = 0 | y_i = 0) = 1-p, 
		\end{aligned}
	\eeq
	where $p< 0.5$. In other words, we flip the genuine label $y_i$ with probability $p$. Then we have
	\[
		\begin{aligned}
		\EE (\epsilon_i\bx_i) & = \EE (\epsilon_i \bx_i 1_{\{y_i=0\}}) + \EE (\epsilon_i \bx_i 1_{\{y_i = 1\}}) = p\EE (\bx_i(1_{\{y_i=0\}} - 1_{\{y_i=1\}})) = 2p\EE (\bx_i 1_{\{y_i=0\}}). 
		\end{aligned}
	\]
	The last equality holds because $\EE \bx_i = \bzero$. Therefore, $\EE (\epsilon_i\bx_i) \propto p$ and if $p$ does not decay, neither does $\EE(\epsilon_i\bx_i)$. \citet{NDR13} solve this noisy label problem through minimizing weighted negative log-likelihood
	\beq
		\label{eq:3.3}
		\begin{aligned}
		& \widehat\bbeta^w := \argmin_{\bbeta \in \RR^d} \frac{1}{n} \sum\limits_{i=1}^n \ell^w(\bx_i, z_i; \bbeta)  = \argmin_{\bbeta\in \RR^d} \frac{1}{n} \sum\limits_{i=1}^n  \frac{(1-p)\ell(\bx_i, z_i; \bbeta) - p \ell(\bx_i, 1-z_i ; \bbeta)}{1 - 2p}. 
		\end{aligned}
	\eeq
	Lemma 1 therein shows that $\EE_{\epsilon_i} \ell^w(\bx_i, z_i) = \ell(\bx_i, y_i)$. This implies that  when the sample size is sufficiently large, minimizing the weighted negative log-likelihood above is similar to minimizing the negative log-likelihood with true labels. In the presence of heavy-tailed features, we propose to replace $\bx_i$ with the $\ell_4$-norm shrunk feature $\widetilde \bx_i$, i.e., we use
	\beq
		\label{eq:3.4}
		\begin{aligned}
		\widetilde\bbeta^w & := \argmin_{\bbeta\in \RR^d} \frac{1}{n}\sum\limits_{i=1}^n \ell^w(\widetilde\bx_i, z_i; \bbeta) = \frac{1}{n}\sum\limits_{i=1}^n \frac{(1-p)\ell(\widetilde\bx_i, z_i; \bbeta) - p \ell(\widetilde\bx_i, 1- z_i; \bbeta)}{1 - 2p}
		\end{aligned}
	\eeq
	to estimate the regression vector $\bbeta^*$. The following corollary establishes the statistical error rate of $\widetilde\bbeta^w$ with an exponential deviation bound. 
	
	\begin{cor}
		\label{cor:1}
		Under the logistic regression with random corruption $\epsilon_i$ satisfying \eqref{eq:3.2}, choose $\tau_1\asymp (n / \log n)^{1/4}$. Under the conditions of Lemma \ref{lem:1}, it holds for some constant $C$ and any $\xi > 1$ such that when $(d \log d/n) ^ {1 / 2}$ is sufficiently small, 
		\[
			\PP\biggl \{ \ltwonorm{\widetilde\bbeta^w - \bbeta^*} \ge C \xi \biggl ( \frac{ d\log n}{n}\biggr ) ^ {1 / 2}\biggr \} \le 2n^{1-\xi}. 
		\]	
	\end{cor}
	
	\begin{rem}
		Here we do not need to truncate the response by $\tau_2$ because in logistic regression the response is always bounded. 
	\end{rem}
	\subsection{High-dimensional regime}
	
	In this section, we consider the regime where the dimension $d$ grows much faster than the sample size $n$. Recall that the standard $\ell_1$-regularized MLE of the regression vector $\bbeta^*$ under the GLM is
	\beq
		\label{eq:3.1}
		\hat\bbeta := \argmin_{\bbeta\in \RR^d} 
		\frac{1}{n}\sum\limits_{i=1}^n \bigl(-y_i\bx_i^\top\bbeta+b(\bx_i^\top\bbeta)\bigr) + \lambda\lonenorm{\bbeta},
	\eeq
	where $(y_i, \bx_i)$ comes from the GLM \eqref{eq:2.1} and $\lambda>0$ is a tuning parameter. \citet{NRW12} show that $\|\widehat\bbeta-\bbeta^*\|_2 = O_{\PP}\{(s\log d/n) ^ {1 / 2}\}$ under the GLM when $\{\bx_i\}_{i=1}^n$ are sub-Gaussian. However, in the presence of heavy-tailed features $\bx_i$ and corruption $\epsilon_i$, the statistical accuracy of $\hat\bbeta$ might deteriorate if we directly evaluate the log-likelihood \eqref{eq:3.1} on $\{(z_i, \bx_i)\}_{i=1}^n $. Our goal is to develop a robust $\ell_1$-regularized MLE for $\bbeta^*$. Let $\widetilde \bx_i$ be the elementwise shrunk version of $\bx_i$ such that for any $j \in [d]$, 
	\[
		\widetilde x_{ij}:= \min (|x_{ij}|, \tau_1) x_{ij}/|x_{ij}|. 
	\] 
	Construct $\widetilde z_i$ as in \eqref{eq:clip_response}. We propose the following $\widetilde\bbeta$ that minimizes the negative log-likelihood on the shrunk data with $\ell_1$-norm regularization: 
	\[
		\widetilde\bbeta := \argmin_{\bbeta\in \RR^d}  \widetilde\ell_n(\bbeta)+\lambda\lonenorm{\bbeta}, 
	\]
	where $\widetilde \ell_n(\bbeta)$ is defined as in \eqref{eq:2.4}, and where $\lambda$ is a tuning parameter. For $\cS \subset [d]$ and $|\cS|=s$, define the restricted cone $\cC(\cS):=\{\bv \in \RR^d: \lonenorm{\bv_{\cS^c}} \le
	3\lonenorm{\bv_{\cS}}\}$. By Lemma 1 in \citet{NRW12}, when $\lambda > 2\supnorm{\nabla \widetilde \ell_n(\bbeta)}$, $\widetilde \bbeta - \bbeta^* \in \cC(\cS)$, which is a crucial property that gives rise to statistical consistency of $\widetilde \bbeta$ under high-dimensional regimes. Therefore, in the following we first present a lemma that characterizes the order of $\|\nabla_{\bbeta} \widetilde\ell_n(\bbeta^*)\|_{\max}$. 
	\begin{lem}
		\label{lem:2}
		Under the following conditions: (1) $\forall i \in [n]$, $b''(\bx_i^\top \bbeta^*) \le M< \infty$ and $\forall\ \omega >0,\ \exists \ m(\omega)>0$ such that $b''(\eta) \ge m(\omega) > 0$ for $|\eta| \le \omega$; (2) $\EE x_{ij}=0$, $\EE x^2_{ij}x^2_{ik} \le R < \infty$ for all $1 \le j, k\le d$; (3) $\EE z_i^4 \le M_1$ and $\EE \epsilon_i^4 \le M_1$; (4) $\lonenorm{\bbeta^*}\le L < \infty$; (5)  $|\EE \epsilon_i x_{ij}| \le M_2/ n ^ {1 / 2}$ for some universal constant $M_2< \infty$ and all $1 \le j \le d$. With $\tau_1, \tau_2 \asymp (n/\log d)^{1/4}$, for any $\xi>1$ we have that
		\[ 
			\PP\biggl \{ \supnorm{\nabla \widetilde \ell(\bbeta^*)}\ge C\xi \biggl ( \frac{\log d}{n}\biggr ) ^ {1 / 2}\biggr \}  \le 2d^{1-\xi}.
		\]
	\end{lem}
%
	Another requirement for the statistical guarantee of $\widetilde \bbeta$ is the restricted strong convexity (RSC) of $\widetilde \ell_n$, which is first formulated in \citet{NRW12}. RSC ensures that $\widetilde \ell_n(\bbeta)$ is ``not too flat'', so that if $|\widetilde\ell_n(\widetilde \bbeta) - \widetilde \ell_n(\bbeta^*)| $ is small, then $\widetilde\bbeta$ and $\bbeta^*$ are close. In high-dimensional sparse linear regression, RSC is implied by the restricted eigenvalue (RE) condition (\citet{BRT09}, \citet{van07}, etc.), a widely studied and acknowledged condition for statistical error analysis of the Lasso estimator. Unlike the quadratic loss in linear regression, the negative log-likelihood $\widetilde \ell_n(\bbeta)$ has its hessian matrix $\widetilde\bH_n(\bbeta)$ depend on $\bbeta$, which creates technical difficulty of verifying its RSC. Here we establish localized RSC (LRSC) of $\tilde \ell(\bbeta)$, i.e., RSC with $\bbeta$ constrained within a small neighborhood of $\bbeta^*$, which has been shown to suffice for statistical analysis of regularized M-estimators in the high-dimensional regime (\citet{FLS17}, \citet{SZF17}). Formally, we say a loss function $\cL(\bbeta)$ satisfies LRSC($\bbeta^*, r, \cS, \kappa, \tau_{\cL}$) if for any $\bDelta \in \cC(\cS) \cap \cB_2(\bzero, r)$, 
	\[
		\delta\cL(\bbeta^* + \bDelta; \bbeta^*) \ge \kappa \ltwonorm{\bDelta}^2 - \tau_{\cL},
	\]
	where $\tau_{\cL}$ is a small tolerance term. The following lemma establishes the LRSC of $\widetilde\ell_n(\bbeta)$. 

	\begin{lem}
		\label{lem:3}
		Suppose the conditions of Lemma \ref{lem:2} hold. Let $\cS$ be the true support of $\bbeta^*$ with $|\cS| = s$. Assume that for any $\bv \in \RR^d$ such that $\bv \in \cC(\cS)$ and $\ltwonorm{\bv} = 1$, $0<\kappa_0  \le \bv^\top\EE (\bx_i\bx_i^\top)\bv \le \kappa_1< \infty$. Set $\tau_1 \asymp (n / \log d)^{1 / 4}$. For any $0<r<1$ and $t>0$, as long as $s^2\log d/ n$ is sufficiently small, we have with probability at least $1 - 2\exp(-t)$ that for any $\bDelta \in \cC(\cS) \cap \cB_2(\bzero, r)$, 
		\[
			\begin{aligned}
			 \delta \widetilde\ell_n(\bbeta^* & + \bDelta; \bbeta^*)  \ge \kappa\ltwonorm{\bDelta}^2 - C_0  r^2 \biggl \{ \biggl ( \frac{t}{n}\biggr ) ^ {1 / 2} + \biggl ( \frac{s \log d}{n}\biggr ) ^ {1 / 2}\biggr \},
			\end{aligned}
		\]
		where $\kappa$ and $C_0$ are constants. 
	\end{lem}
	
	Combining Lemmas \ref{lem:2} and \ref{lem:3} yields the statistical guarantee of $\widetilde\bbeta$ as follows. 
	
	\begin{thm}
		\label{thm:2}
		Under the assumptions of Lemma \ref{lem:2} and \ref{lem:3}, choose $\lambda = 2C\xi(\log d / n) ^ {1 / 2}$ and $\tau_1, \tau_2 \asymp (n/ \log d)^{1/4}$, where $\xi$ and $C$ are the same constants as in Lemma \ref{lem:2}. Then there exists a constant $C_1 > 0$ such that 
		\[
			\PP\biggl \{  \ltwonorm{\widetilde\bbeta - \bbeta^*} \ge C_1\xi \biggl ( \frac{s \log d}{n}\biggr ) ^ {1 / 2} \biggr \} \le 4d^{1- \xi}. 
		\]
	\end{thm}

	\section{Numerical study}
	
		\begin{figure*}[t]
		\centering
		\vspace{-1.4cm}
		\label{fig:2}
		\begin{tabular}{cc}
			\includegraphics[scale=.35]{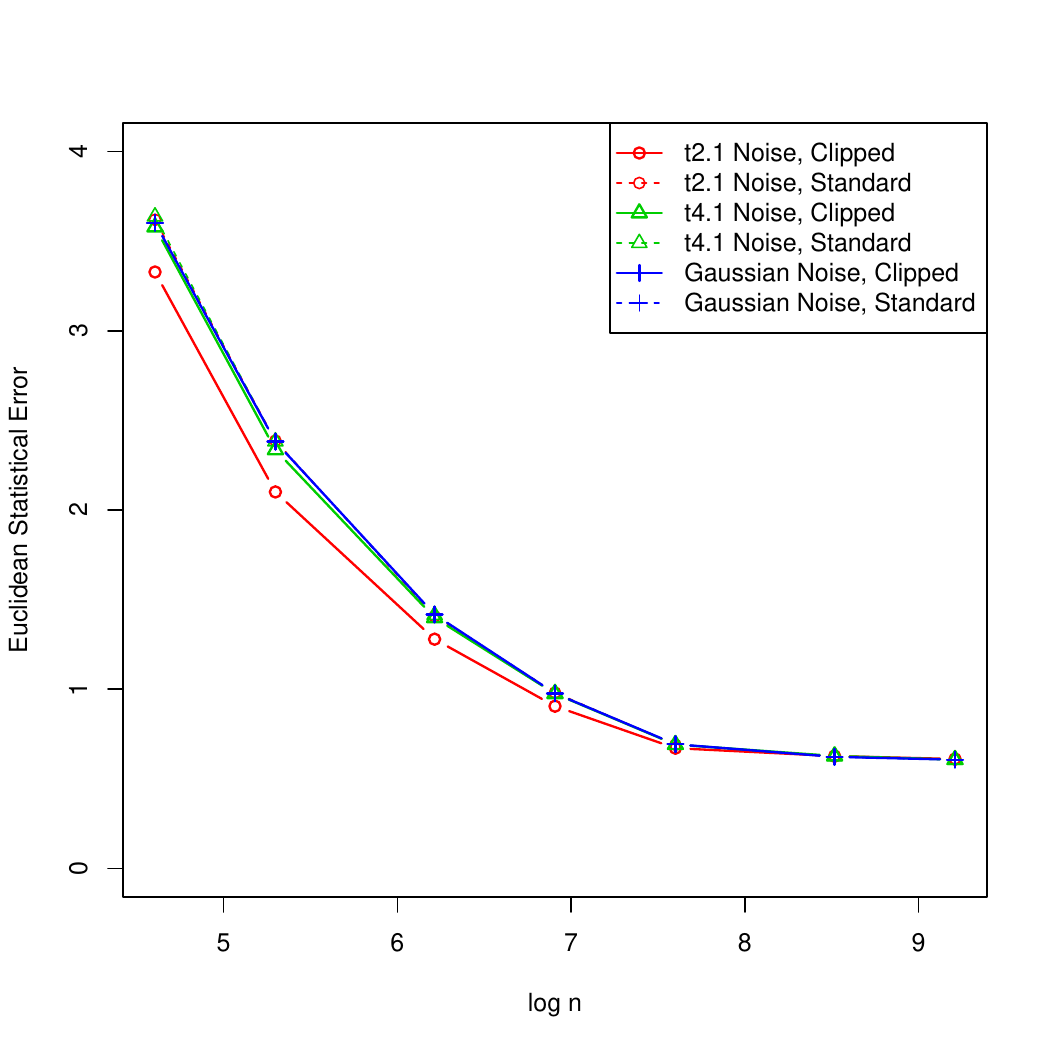} & \includegraphics[scale=.35]{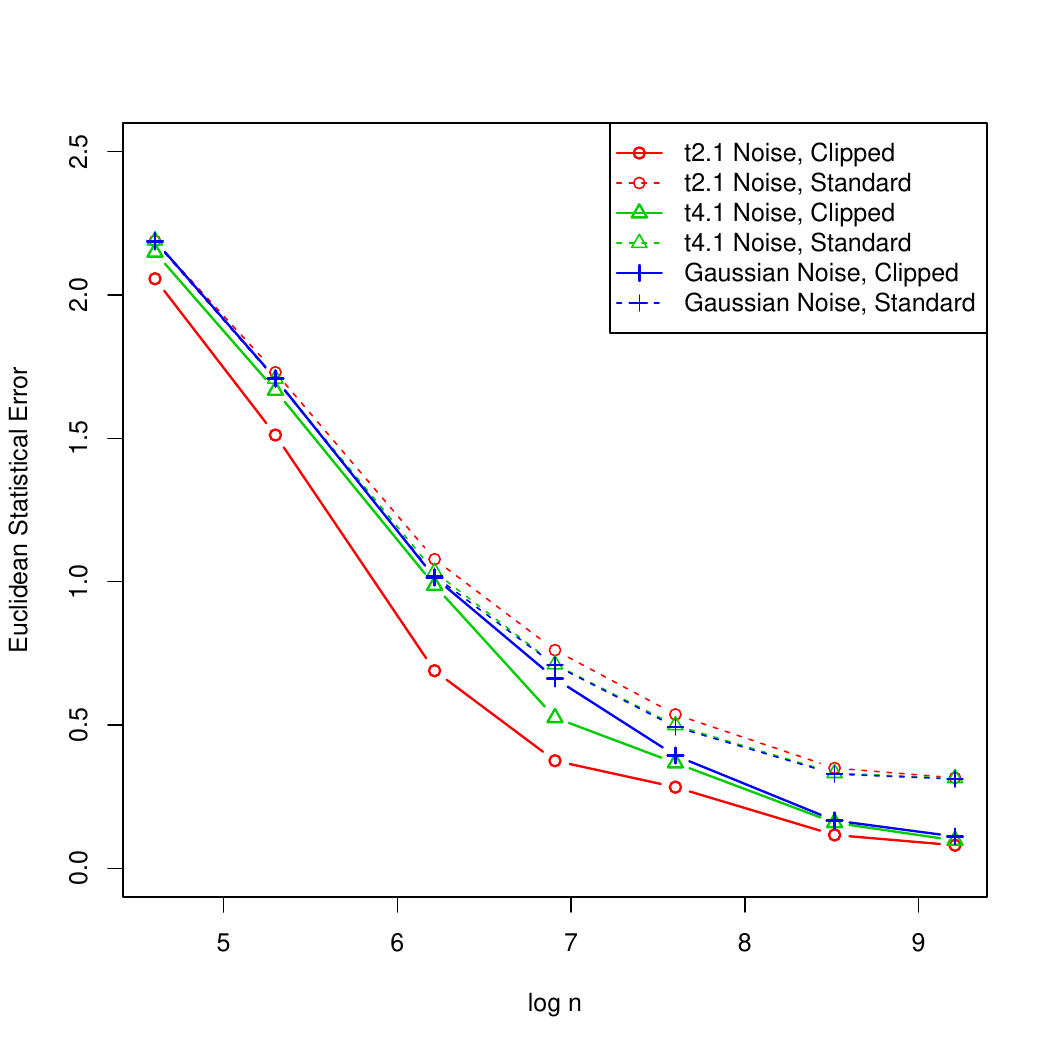} \\
			\quad Standard Gaussian features & \quad $t_{4.1}$ features
		\end{tabular}
		\caption{High dimensional sparse linear regression with light-tailed features (left) and heavy-tailed features (right)}
		\vspace{-.4cm}
	\end{figure*}

	\subsection{High-dimensional sparse linear regression}
	
	We first consider the high-dimensional sparse linear model $y_i = \bx_i^\top\bbeta^* + \epsilon_i$. We set $d=1000$, $n = 100, 200, 500, 1000, 5000, 10000$ and $\bbeta^* = (1, 1, 1, 1, 1, 0, \ldots, 0)^\top$. Recall that in the high-dimensional regime, we propose elementwise shrinkage on the heavy-tailed features and clip the responses. In Figure 2, we compare estimation error of the $\ell_1$-regularized least squares estimators based on the shrunk data and original data under standard Gaussian features and $t_{4.1}$ features respectively. All feature vectors $\{\bx_i\}_{i=1}^n$ are i.i.d., and within each $\bx_i$, $\{x_{ij}\}_{j=1}^d$ are i.i.d. $\{\epsilon_i\}_{i=1}^n$ are i.i.d. noises that are independent of the features and we adjust the magnitude of the noise such that $\mathrm{SD}(\epsilon_i) =5$ regardless of its distribution. $\tau_1, \tau_2$ and  $\lambda$ are selected by cross-validation. The plot is based on $1,000$ independent Monte Carlo simulations. From Figure 2, we first observe that under both light-tailed and heavy-tailed features, the heavier tail $\epsilon_i$ has, the more the data shrinkage approach improves the statistical accuracy. More importantly, the benefit from data shrinkage is much more significant in the presence of heavy-tailed features, which justifies our theory. 
	
\begin{figure*}[t]
		\centering
		\vspace{-.5cm}
		\begin{tabular}{cc}
			\includegraphics[ scale=.34]{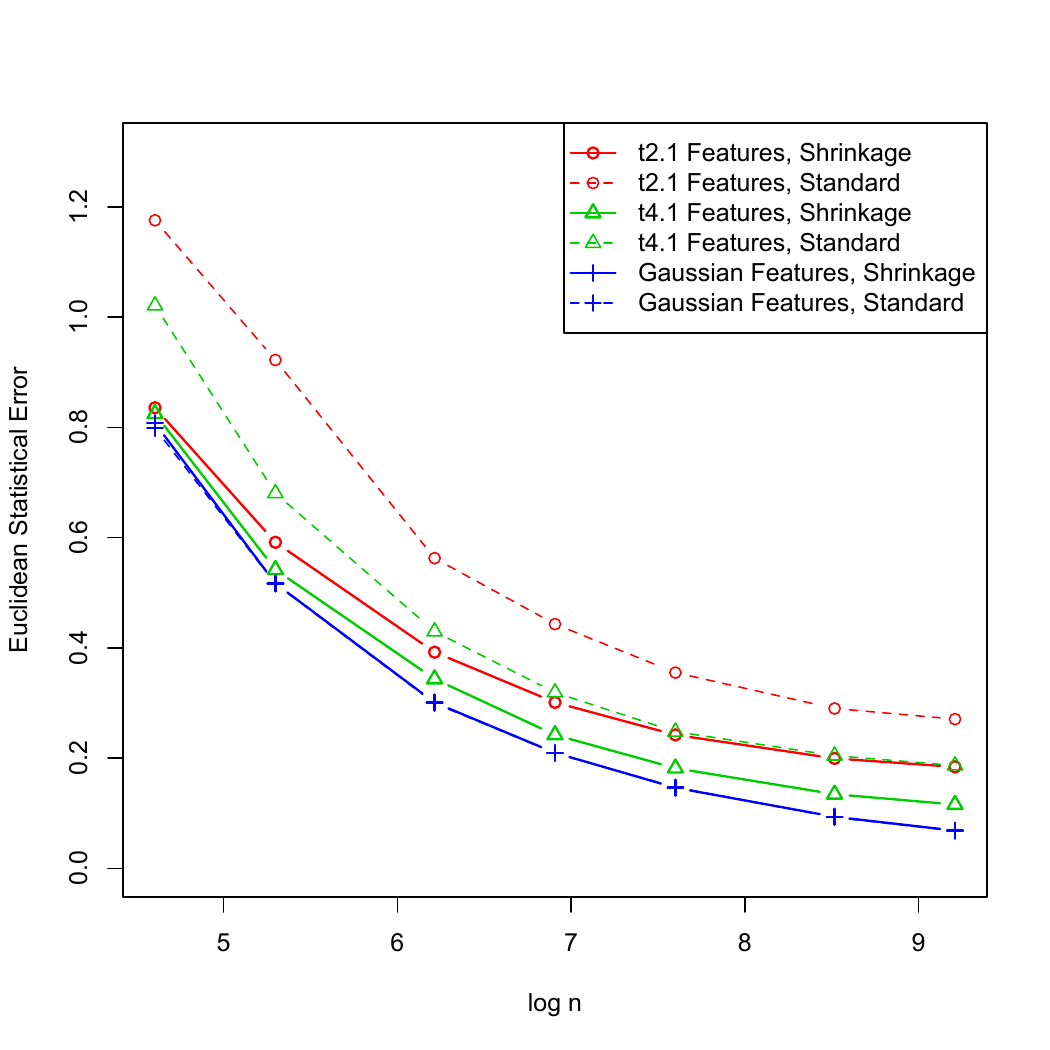} & \includegraphics[scale=.34]{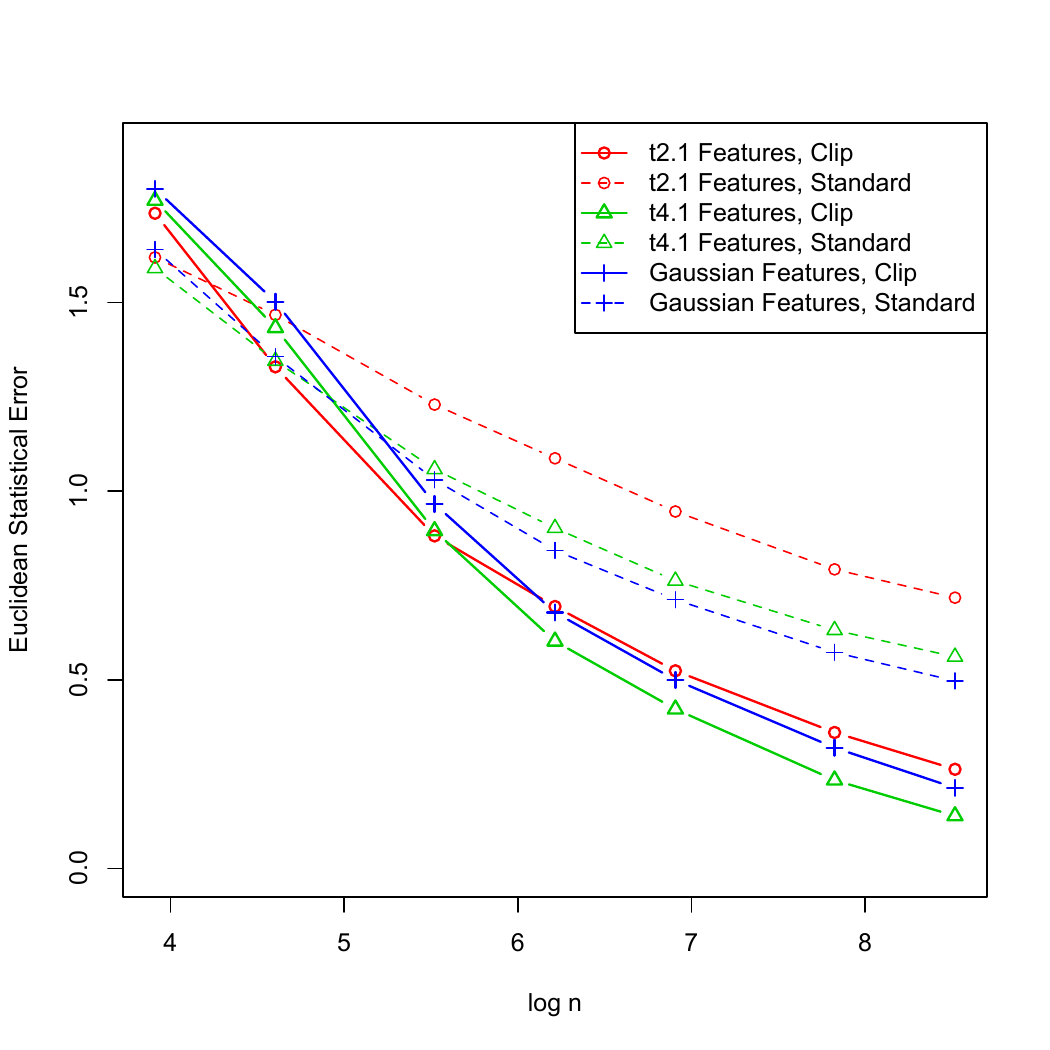} \\
			\quad Low dimensions & \quad High dimensions
		\end{tabular}
		\vspace{-.1cm}	
		\caption{Statistical error of the MLEs based on minimizing $\widetilde\ell^w_n(\bbeta)$ with $10\%$ mislabeled data}
		\vspace{-.5cm}
		\label{fig:3}
	\end{figure*}

	\subsection{Logistic regression with mislabeled data}
	
	In this subsection, we consider the logistic regression with mislabeled data as characterized by \eqref{eq:3.2}. We minimize the weighted negative log-likelihood to derive $\widehat\bbeta^w$ and $\widetilde\bbeta^w$ as described in \eqref{eq:3.3} and \eqref{eq:3.4} to estimate the regression vector $\bbeta^*$ and compare their performance. The tuning parameters $\lambda$ and $\tau_1$ are chosen based on cross-validation. We investigate both the low-dimensional and high-dimensional regimes. 

	In the low-dimensional regime, let $d=10$, $n$ range from $10 ^ 2$ to $10 ^ 4$, $\bbeta^* = (0.5{\bf 1}^{\top}_5, -0.5{\bf 1}^{\top}_5)^{\top}$ and $p=0.1$. The left panel of Figure \ref{fig:3} compares $\ltwonorm{\widehat\bbeta^w - \bbeta^*}$ and $\ltwonorm{\widetilde\bbeta^w - \bbeta^*}$ under $t_{2.1}, t_{4.1}$ and Gaussian features. We can observe that $\widetilde\bbeta^w$ significantly outperforms $\widehat\bbeta^w$ under $t_{2.1}$ and $t_{4.1}$ features, and they perform equally well when features are Gaussian. This perfectly validates our theory. In the high-dimensional regime, we apply elementwise shrinakge to $\bx_i$ to derive $\widetilde\bbeta^w$. Let $d=100$, $n$ range from $50$ to $5,000$, $\bbeta^* = (1, 1, -1, {0, \ldots, 0})$ and $p=0.1$. As shown in the right panel of Figure \ref{fig:3}, $\widetilde\bbeta^w$ enjoys sharper statistical accuracy than $\widehat\bbeta^w$ under all the three types of features. The outstanding performance of $\widetilde \bbeta^w$ under the Gaussian feature scenario is particularly surprising. We conjecture that feature shrinkage here downsizes $\supnorm{\nabla \widetilde \ell_n^w(\bbeta^*)}$ and thus leads to more effective regularization. 
	
	\subsection{Experiments on the MNIST dataset}
	
	Motivated by the effectiveness of feature shrinkage, we incorporate a shrinkage layer to a convolutional neural network (CNN) to robustify its classification performance on corrupted images. Figure \ref{fig:4} illustrates this new architecture, which we call a shrinkage CNN. The new shrinkage layer applies the $\ell_4-$norm shrinkage as in \eqref{eq:l4_norm_shrinkage} to the feature vector $\bx$ learned by the original CNN to guard against its heavy tail if any. Then the shrunk features are used to derive the probability of each class.  
	
	\begin{figure}[H]
		\centering
		\vspace{-.4cm}
		\includegraphics[width=.9 \textwidth]{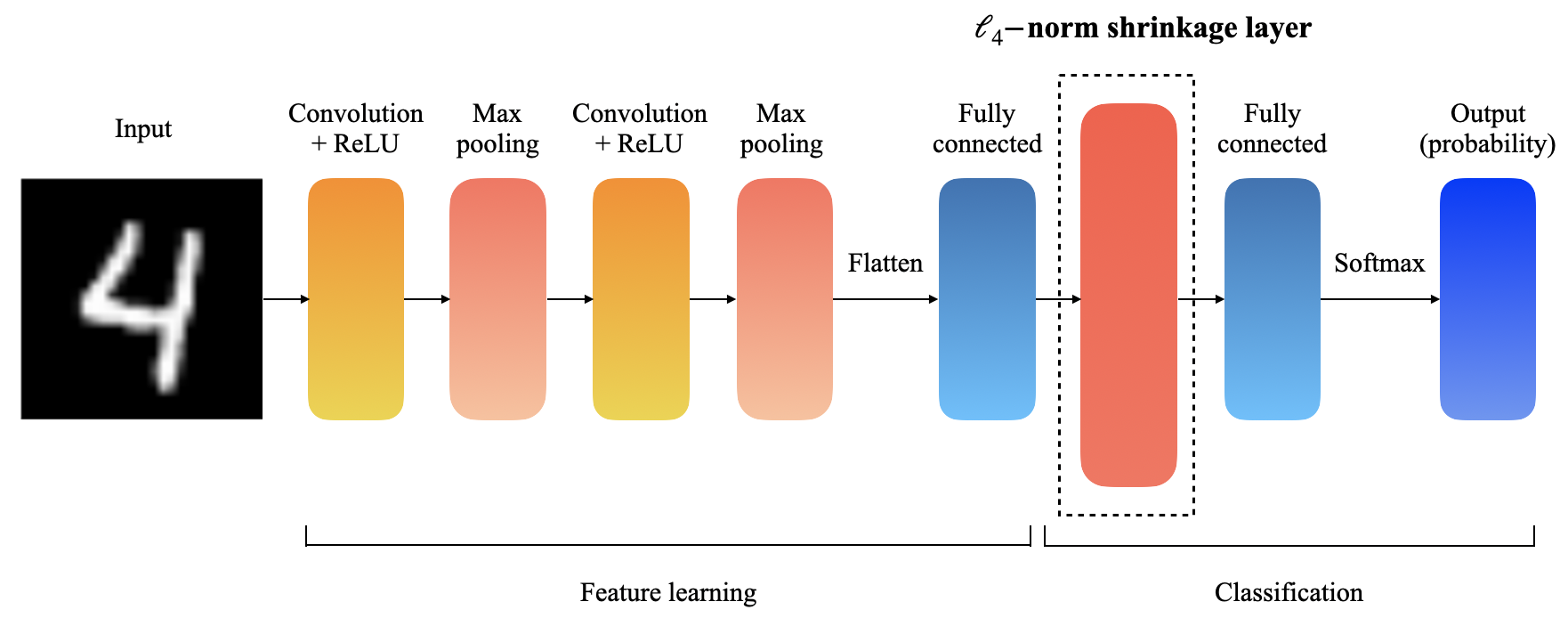} 
		\vspace{-.3cm}
		\caption{Architecture of the shrinkage CNN}
		\label{fig:4}
		\vspace{-.6cm}
	\end{figure}
	
	We classify the digits $4$'s and $9$'s in the MNIST (\citet{LeC98}) dataset when the images are randomly mislabeled with probability $0.4$ and corrupted by ``salt'' noise. We train both the original CNN and its shrinkage variant by minimizing the weighted negative log-likelihood $\widetilde \ell^w_n(\bbeta)$ in \eqref{eq:3.4}. We choose $\tau_1 = 2$ in \eqref{eq:l4_norm_shrinkage} in the $\ell_4$-norm shrinakge layer. We repeat flipping labels, adding noise and training for $100$ times independently to evaluate the average misclassification rate. The result is presented in Table \ref{tab:2}. We can see that the feature shrinkage layer reduces the testing misclassification rate by more than $30\%$ relatively in the presence of noisy pixels.
%
%
%

\begin{table}[H]
  \caption{Average testing misclassification rate (with standard error in the parentheses) on noisy MNIST images under mislabeling probability $40\%$}
  \label{tab:2}
  \centering
  \begin{tabular}{ccc}
    \toprule
    Proportion of noisy pixels & Original CNN & Shrinkage CNN \\
    \midrule
    0        & $3.64\%_{(0.20\%)}$  & $2.93\%_{(0.09\%)}$ \\ 
    0.1     & $6.88\%_{(0.22\%)}$  & $4.18\%_{(0.17\%)}$ \\
    0.2     & $6.90\%_{(0.21\%)}$  & $4.37\%_{(0.16\%)}$ \\
    0.4     & $10.69\%_{(0.29\%)}$ & $6.65\%_{(0.24\%)}$ \\
    0.6     & $18.82\%_{(0.88\%)}$ & $12.80\%_{(0.65\%)}$\\
    
    
    \bottomrule
  \end{tabular}
\end{table}

\bibliographystyle{ims}
\bibliography{RobustBib}

\newpage


\section{Technical lemmas, propositions and proofs}
	
	\begin{lem}
		\label{lem:4}
		Suppose $\EE(\bv^\top\bx_i)^4\le R$ for any $\bv\in \cS^{d-1}$. Define the $\ell_4$-norm shrunk samples 
		\[
		 	\widetilde \bx_i := \frac{\min(\lfournorm{\bx_i}, \tau)}{\lfournorm{\bx_i}} \bx_i, 
		\]
		where $\tau$ is a threshold value. Then we have the following:
		\begin{enumerate}
			\item $\opnorm{\widetilde\bx_i\widetilde\bx_i^\top- \EE \widetilde\bx_i\widetilde\bx_i^\top} \le \ltwonorm{\widetilde\bx_i}^2+ \sqrt{R} \le \sqrt{d}\tau^2+\sqrt{R};$
			\item $\opnorm{\EE((\widetilde\bx_i\widetilde\bx_i^\top- \EE\widetilde\bx_i\widetilde\bx_i^\top)^\top(\widetilde\bx_i\widetilde\bx_i^\top-\EE \widetilde\bx_i\widetilde\bx_i^\top))} \le R(d+1);$
			\item For all $\xi > 0$, $\PP\biggl \{ \opnorm{\widetilde\bSigma_n(\tau) - \bSigma} \ge \xi \biggl ( \frac{Rd \log n}{n}\biggr ) ^ {1 / 2}\biggr \}\le n^{1-C\xi}, $ where $\tau \asymp  \bigl(nR/(\log n) \bigr)^{1/4}$ and $C$ is a universal constant.
		\end{enumerate}

	\end{lem}
	
	\begin{proof}
		This result is from \cite{FWZ17}. For convenience of adapting the lemma to other settings, we present its proof here. Notice that 
		\beq
			\label{eq:5.1}
			\opnorm{\widetilde\bx_i\widetilde\bx_i^\top- \EE\widetilde\bx_i\widetilde\bx_i^\top}\le \opnorm{\widetilde \bx_i \widetilde\bx_i^\top}+\opnorm{\EE \widetilde\bx_i\widetilde\bx_i^\top}=\ltwonorm{\widetilde\bx_i}^2+ \sqrt{R} \le \sqrt{d}\tau^2+\sqrt{R}. 
		\eeq
		Also for any $\bv\in \cS^{d-1}$, we have
		\[	
			\begin{aligned}
				\EE(\bv^\top \widetilde\bx_i\widetilde \bx_i^\top\widetilde\bx_i\widetilde\bx_i^\top \bv) & =\EE (\ltwonorm{\widetilde\bx_i}^2 (\bv^\top\widetilde\bx_i)^2)\le \EE  (\ltwonorm{\bx_i}^2 (\bv^\top\bx_i)^2) \\
			 	& = \sum\limits_{j=1}^d \EE (x_{ij}^2(\bv^\top \bx_i)^2) \le \sum\limits_{j=1}^d \sqrt{\EE (x_{ij}^4) \EE(\bv^\top \bx_i)^4 } \le Rd
			\end{aligned}
		\]
		Then it follows that $\opnorm{\EE\widetilde\bx_i\widetilde\bx_i^\top\widetilde\bx_i\widetilde\bx_i^\top}\le Rd$. Since $\opnorm{(\EE \widetilde \bx_i\widetilde \bx_i^\top)^\top\EE (\widetilde\bx_i\widetilde\bx_i^\top)}\le R$, 
		\beq
			\label{eq:5.3}
			\opnorm{\EE((\widetilde\bx_i\widetilde\bx_i^\top- \EE\widetilde\bx_i\widetilde\bx_i^\top)^\top(\widetilde\bx_i\widetilde\bx_i^\top-\EE\widetilde\bx_i\widetilde\bx_i^\top))} \le R(d+1).
		\eeq
		By the matrix Bernstein's inequality (Theorem 5.29 in \cite{Ver10}), we have for some constant $c_1$,
		\beq
			\label{eq:5.2}
			\PP \Big( \opnorm{\frac{1}{n}\sum\limits_{i=1}^n \widetilde\bx_i\widetilde\bx_i^\top-\EE \widetilde\bx_i\widetilde\bx_i^\top}>t \Big)\le 2d\exp\Bigl(-c_1\bigl(\frac{nt^2}{R(d+1)} \wedge \frac{nt}{\sqrt{d}\tau^2+\sqrt{R}}\bigr)\Bigr).	
		\eeq
		For any $\bv\in \cS^{d-1}$, it holds that
		\beq
			\label{eq:5.4}
			\begin{aligned}
				\EE (\bv^\top(\bx_i\bx_i^\top)\bv 1_{\{\lfournorm{\bx_i} \ge \tau\}}) & \le  \sqrt{\E (\bv^\top\bx_i)^4 P(\|\bx_i\|_4>\tau)}  \le \biggl ( \frac{R^2d}{\tau^4}\biggr ) ^ {1 / 2}=\frac{R\sqrt{d}}{\tau^2}.
			\end{aligned}
		\eeq
		Therefore we have 
		\beq
			\label{eq:5.5}
			\opnorm{\EE (\bx_i\bx_i^\top - \widetilde\bx_i\widetilde\bx_i^\top)} \le R\sqrt{d}/\tau^2.
		\eeq
		
		Choose $\tau \asymp (nR/ \log d)^{1/4}$ and substitute $t$ with $\xi\sqrt{Rd\log n/n}$. Then we reach the final conclusion by combining the concentration bound and bias bound.  

	\end{proof}

	\begin{proof}[Proof of Lemma \ref{lem:1}]
	
		Define a contraction function 
		\[
			\phi(x; \theta) = x^2  \ind_{\{ |x| \le \theta\}} + (x - 2\theta)^2 \ind_{\{ \theta < x \le 2\theta\}} + (x + 2\theta)^2 \ind_{\{ -2\theta \le x < -\theta\}}. 
		\]
		One can verify that $\phi(x; \theta) \le x^2$ for any $\theta$. This contraction function was used in a preliminary version of \cite{NRW12} to establish the RSC of negative log-likelihood. Given any $\bDelta \in \cB_2(\bzero, r)$, by the Taylor expansion, we can find $v \in (0, 1)$ such that
		\beq
			\label{eq:5.6}
			\begin{aligned}
				& \delta\widetilde \ell_n(\bbeta^* + \bDelta; \bbeta^*)  = \widetilde \ell_n(\bbeta^* + \bDelta) - \widetilde \ell_n(\bbeta^*) - \nabla\widetilde\ell_n(\bbeta^*)^\top\bDelta = \frac{1}{2} \bDelta^\top \widetilde\bH_n(\bbeta^* + v\bDelta)\bDelta \\
				& = \frac{1}{2n}\sum\limits_{i=1}^n b''(\widetilde \bx_i^\top (\bbeta^* + v\bDelta)) (\bDelta^\top \widetilde\bx_i)^2 \ge \frac{1}{2n}\sum\limits_{i=1}^n b''(\widetilde \bx_i^\top (\bbeta^* + v\bDelta)) \phi(\bDelta^\top \widetilde\bx_i; \alpha_1 r)  \ind_{\{ |{\bbeta^*}^\top \widetilde\bx_i| \le \alpha_2 \} } \\
				& \ge \frac{m(\omega)}{2n}\sum\limits_{i=1}^n \phi(\bDelta^\top \widetilde\bx_i; \alpha_1 r) \ind_{\{ |{\bbeta^*}^\top \widetilde\bx_i| \le \alpha_2 \} }, 
			\end{aligned}
		\eeq
		where we choose $\omega = \alpha_1 + \alpha_2 > \alpha_1 r + \alpha_2$ so that the last inequality holds by Condition (1) . For ease of notation, let $\cA_i := \{|\bDelta^\top\widetilde \bx_i| \le \alpha_1 r\}$ and $\cB_i := \{|{\bbeta^*}^\top\widetilde \bx_i| \le \alpha_2\}$. We have
		\[
			\begin{aligned}
				\EE [\phi(\bDelta^\top \widetilde\bx_i; \alpha_1 r)  \ind_{\cB_i}]& \ge \EE [(\bDelta^\top \widetilde \bx_i)^2  \ind_{\cA_i \cap \cB_i}] \\
			 & \ge \bDelta^\top \EE [\bx_i\bx_i^\top  \ind_{\cA_i \cap \cB_i}] \bDelta  - \bDelta^\top\EE [(\bx_i\bx_i^\top - \widetilde\bx_i\widetilde\bx_i^\top) \ind_{\cA_i \cap \cB_i}] \bDelta \\
			& \ge \bDelta^\top \EE [\bx_i\bx_i^\top \ind_{\cA_i \cap \cB_i}] \bDelta  - \bDelta^\top\EE [\bx_i\bx_i^\top - \widetilde\bx_i\widetilde\bx_i^\top] \bDelta \\
			& \ge \bDelta^\top \EE (\bx_i\bx_i^\top)\bDelta - \bDelta^\top\EE (\bx_i\bx_i^\top \ind_{\cA^c_i \cup \cB^c_i})\bDelta - \bDelta^\top\EE [\bx_i\bx_i^\top - \widetilde\bx_i\widetilde\bx_i^\top] \bDelta \\
			& \ge \kappa_0 \ltwonorm{\bDelta}^2 - \sqrt{\EE (\bDelta^\top \bx_i) ^4 (\PP(\cA_i^c) + \PP(\cB_i^c))} - \bDelta^\top\EE [\bx_i\bx_i^\top - \widetilde\bx_i\widetilde\bx_i^\top] \bDelta. 
			\end{aligned}
		\]
		By the Markov Inequality, 
		\[
			\PP(\cA^c_i) \le \frac{\EE (\bDelta^\top \widetilde \bx_i)^4}{\alpha^4_1 r^4} \le \frac{R}{\alpha^4_1}\quad\text{and}\quad \PP(\cB^c_i) \le \frac{\EE ({\bbeta^*}^\top \widetilde \bx_i)^4}{\alpha^4_2} \le \frac{R\ltwonorm{\bbeta^*}^4}{\alpha^4_2} \le \frac{RL^4}{\alpha^4_2}. 
		\]
		Besides, according to \eqref{eq:5.5}, 
		\[
			\bDelta^\top\EE [\bx_i\bx_i^\top - \widetilde\bx_i\widetilde\bx_i^\top] \bDelta \le \frac{R\sqrt{d}\ltwonorm{\bDelta}^2}{\tau_1^2} \le C_1R\ltwonorm{\bDelta}^2 \biggl ( \frac{d\log d}{n}\biggr ) ^ {1 / 2}, 
		\]
		where $C_1$ is certain constant. Therefore, for sufficiently large $\alpha_1, \alpha_2, n$ and $d$,  
		\beq
			\label{eq:5.7}
			\EE [\phi(\bDelta^\top \widetilde\bx_i; \alpha_1 r) \ind_{\cB_i}] \ge \frac{\kappa_0}{2} \ltwonorm{\bDelta}^2. 
		\eeq
		For notational convenience, define $Z_i := \phi(\bDelta^\top \widetilde\bx_i; \alpha_1 r)  \ind_{ \cB_i } = \phi(\bDelta^\top \widetilde \bx_i  \ind_{ \cB_i }; \alpha_1 r)$ and $\Gamma_r := \sup_{\ltwonorm{\bDelta} \le r} \bigl | n^{-1} \sum\limits_{i=1}^n Z_i - \EE Z_i \bigr|$. Then an application of Massart's inequality (\cite{Mas00}) delivers that
		\beq
			\label{eq:5.8}
			\PP\biggl \{ |\Gamma_r - \EE \Gamma_r | \ge \alpha_1^2r^2\biggl ( \frac{t}{n} \biggr ) ^ {1 / 2}\biggr \} \le 2\exp\bigl( - \frac{t}{8}\bigr). 
		\eeq
		The remaining job is to derive the order of $\EE \Gamma_r$. Note that $|\phi(x_1; \theta) - \phi(x_2; \theta)| \le 2\theta |x_1 - x_2|$ for any $x_1, x_2 \in \RR$. By the symmetrization argument and then Ledoux-Talagrand contraction inequality (see \cite{LTa13}, p. $112$), for a sequence of i.i.d. Rademacher variables $\{\gamma_i\}_{i=1}^n$, 
		\[
			\begin{aligned}
			\EE \Gamma_r & \le 2\EE \sup_{\ltwonorm{\bDelta} \le r} \bigl | \frac{1}{n} \sum\limits_{i=1}^n \gamma_iZ_i \bigr | \le 8\alpha_1 r  \EE \sup_{\ltwonorm{\bDelta} \le r} \bigl| \inn{\frac{1}{n}\sum\limits_{i=1}^n \gamma_i \widetilde \bx_i  \ind_{\{|{\bbeta^*}^\top \widetilde\bx_i| \le \alpha_2\}}, \bDelta}  \bigr| \\
			& \le 8\alpha_1 r^2 \EE \ltwonorm{ \frac{1}{n}\sum\limits_{i=1}^n \gamma_i \widetilde \bx_i  \ind_{\{|{\bbeta^*}^\top \widetilde\bx_i| \le \alpha_2\} }} \le 8\alpha_1 r^2  \biggl ( \EE \ltwonorm{ \frac{1}{n}\sum\limits_{i=1}^n \gamma_i \widetilde \bx_i \ind_{\{|{\bbeta^*}^\top \widetilde\bx_i| \le \alpha_2\}} } ^2\biggr ) ^ {1 / 2} \\
			& \le  8\alpha_1 r^2  \biggl ( \frac{1}{n^2} \sum\limits_{i=1}^n \EE \ltwonorm{\widetilde \bx_i}^2 \biggr ) ^ {1 / 2} \le  8\alpha_1r^2R^{1/4} \biggl ( \frac{d}{n}\biggr ) ^ {1 / 2}. 
			\end{aligned}
		\]
		Combining the above inequality with \eqref{eq:5.6}, \eqref{eq:5.7} and \eqref{eq:5.8} yields that for any $t > 0$, with probability at least $1 - 2\exp(-t)$, for all $\bDelta \in \RR^d$ such that $\ltwonorm{\bDelta} \le r$, 
		\[
			\delta \widetilde\ell_n(\bbeta; \bbeta^*) \ge \frac{m\kappa_0}{4}\ltwonorm{\bDelta}^2 - \alpha_1^2\biggl ( \frac{8t}{n}\biggr ) ^ {1 / 2}  r^2 - 8\alpha_1R^{1/4} \biggl ( \frac{d}{n}\biggr ) ^ {1 / 2}  r^2 . 
		\]
	\end{proof}
	
	\begin{proof}[Proof of Theorem \ref{thm:1}]

Construct an intermediate estimator $\widetilde\bbeta_{\eta} $ between $\widetilde\bbeta$ and $\bbeta^*$:
		\[
			\widetilde \bbeta_{\eta} = \bbeta^* + \eta (\widetilde \bbeta - \bbeta^*), 
		\]
		where $\eta = 1$ if $\ltwonorm{\widetilde \bbeta - \bbeta^*} \le r$ and $\eta  = r/ \ltwonorm{\widetilde\bbeta - \bbeta^*}$ if $\ltwonorm{\widetilde\bbeta - \bbeta^*} > r$. Write $\widetilde\bbeta_{\eta} - \bbeta^*$ as $\widetilde \bDelta_{\eta}$. By Lemma \ref{lem:1}, it holds with probability at least $1 - 2\exp(-t)$ that 
	\[
		\kappa \ltwonorm{\widetilde\bDelta_{\eta}}^2 - C r^2\biggl \{ \biggl ( \frac{t}{n}\biggr ) ^ {1 / 2} + \biggl ( \frac{d}{n} \biggr ) ^ {1 / 2}\biggr \} \le \delta\widetilde\ell_n(\widetilde\bbeta_{\eta}; \bbeta^*) \le - \nabla\widetilde\ell_n(\bbeta^*)^\top\widetilde\bDelta_{\eta} \le \ltwonorm{\nabla \widetilde \ell_n(\bbeta^*)} \ltwonorm{\widetilde \bDelta_{\eta}},
	\]
	which further implies that 
	\beq
		\label{eq:5.9}
		\begin{aligned}
			\ltwonorm{\widetilde \bDelta_{\eta}} & \le \frac{ 3\ltwonorm{\nabla\widetilde\ell_n(\bbeta^*)}}{\kappa} + \biggl ( \frac{3c_1 r^2}{\kappa}\biggr ) ^ {1 / 2}  \Bigl( \frac{ t}{n}\Bigr)^{1/4} + \biggl (\frac{3 c_2r^2}{\kappa}\biggr ) ^ {1 / 2} \Bigl( \frac{d}{n}\Bigr)^{1/4}. 
		\end{aligned}
	\eeq
	Now we derive the rate of $\ltwonorm{\nabla \widetilde\ell_n(\bbeta^*)}$.
			\beq
				\begin{aligned}				
				 \nabla \widetilde\ell_n(\bbeta^*) & = \frac{1}{n}\sum\limits_{i=1}^n (\widetilde z_i - b'(\widetilde\bx_i^\top\bbeta^*))\widetilde \bx_{i}  \\
				 & = \underbrace{\frac{1}{n}\sum\limits_{i=1}^n \widetilde z_i \widetilde \bx_i - \EE \widetilde z_i \widetilde\bx_i}_{T_1} + \underbrace{\EE (\widetilde z_i - b'(\widetilde\bx_i^\top\bbeta^*))\widetilde\bx_i }_{T_2} + \underbrace{\frac{1}{n}\sum\limits_{i=1}^n b'(\widetilde\bx_i^\top\bbeta^*)\widetilde\bx_i - \EE (b'(\widetilde\bx_i^\top\bbeta^*)\widetilde\bx_i)}_{T_3}. 
				\end{aligned}
			\eeq
			where $\overline\bx_i$ is between $\bx_i$ and $\widetilde\bx_i$ by the mean value theorem. In the following we will bound $T_1, T_2$ and $T_3$ respectively.
			
			\vspace{.1cm}
			\underline{\emph{Bound for $T_1$}}:\quad Define the Hermitian dilation matrix 
			\[
				\widetilde \bZ_i := \widetilde z_i  \left( \begin{array}{cc}
					0 & \widetilde\bx_i^\top \\
					\widetilde \bx_i & \bzero
				\end{array}\right)
			\]
			Note that 
			\[
				\begin{aligned}
				\opnorm{\EE \widetilde\bZ^{2}_i} & = \opnorm{\EE \Bigl [{\widetilde z_i}^{2} \left( \begin{array}{cc}
					\widetilde\bx_i^\top\widetilde\bx_i & \bzero^\top \\
					\bzero & \widetilde\bx_i \widetilde\bx_i^\top 
				\end{array}\right) \Bigr]} = \max (\EE (\widetilde z_i^2\widetilde\bx_i^\top\widetilde \bx_i), \opnorm{ \EE (\widetilde z_i^2 \widetilde \bx_i\widetilde\bx_i^\top)})
				\end{aligned}
			\]
			For any $j \in [d]$, 
			\[
				\begin{aligned}
					\E \bigl(\widetilde z_i^2  \widetilde x^2_{ij} \bigr) & \le \sqrt{\E z_i^{4}  \E x^4_{ij}} \le \sqrt{M_1R},
				\end{aligned}
			\]
			so $\EE [\widetilde z_i^2  \widetilde\bx_i^\top\widetilde\bx_i] \le d\sqrt{M_1R}$. In addition, for any $\bv\in \RR^d$ such that $\ltwonorm{\bv} =1$, 
			\[
				\E \bigl(\widetilde z_i^2 (\bv^\top\widetilde\bx_i)^2\bigr) \le \sqrt{M_1R}. 
			\]
			We thus have $\opnorm{\EE \widetilde \bZ_i^2} \le d\sqrt{M_1R}$. 
			In addition, $\opnorm{\EE \widetilde \bZ_i} = \EE ( \widetilde z_i \ltwonorm{\widetilde\bx_i}) \le \sqrt{\EE z_i^2 \EE \ltwonorm{\bx_i}^2} \le \sqrt{d}(M_1 R)^{1/4}$, which further implies that $\opnorm{\EE (\widetilde \bZ_i - \EE \widetilde\bZ_i)^{2}} \le (d+1)\sqrt{M_1R}$. 
			Also notice that since $\lfournorm{\widetilde \bx_i} \le \tau_1$ and $\widetilde z_i \le \tau_2$, $\opnorm{\widetilde \bZ_i} \le \frac{1}{2}d^{1/4}  \tau_1 \tau_2.$ By the matrix Bernstein's inequality, 
			\[
				P\bigl( \opnorm{\frac{1}{n} \sum\limits_{i=1}^n \widetilde \bZ_i - \EE \widetilde \bZ_i} \ge t \bigr) \le d  \exp\Bigl( -c_1\min\bigl(\frac{nt^2}{(d+1)\sqrt{M_1R}}, \frac{2nt}{d^{1/4} \tau_1\tau_2}\bigr) \Bigr). 
			\]
			
			Given that $\ltwonorm{T_1} = 2\opnorm{ n^{-1} \sum\limits_{i=1}^n \widetilde\bZ_i - \EE \widetilde \bZ_i}$, it thus holds that 
			\beq
				\label{eq:5.11}
				\PP\Bigl( \ltwonorm{T_1} \ge 2t \Bigr) \le d \exp\Bigl( -c_1\min\bigl(\frac{nt^2}{(d+1)\sqrt{M_1R}}, \frac{2nt}{d^{1/4} \tau_1\tau_2}\bigr) \Bigr). 
			\eeq

			\underline{\emph{Bound for $T_2$}}: \quad We decompose $T_2$ as follows:
			\[
				\begin{aligned}
				\ltwonorm{T_2} & \le \underbrace{\ltwonorm{\EE (\widetilde z_i - z_i)\widetilde \bx_i}}_{T_{21}} + \underbrace{\ltwonorm{\EE (z_i - y_i)\widetilde \bx_i}}_{T_{22}} + \underbrace{\ltwonorm{\EE (y_i - b'(\bx_i^\top\bbeta^*))\widetilde\bx_i}}_{T_{23}} \\
				& + \underbrace{\ltwonorm{\EE (b'(\bx_i^\top\bbeta^*) - b'(\widetilde \bx_i^\top\bbeta^*))\widetilde\bx_i}}_{T_{24}}. 
				\end{aligned}
			\]	 
			
			Now we work on $\{T_{2i}\}_{i=1}^4$ one by one. For any $\bv\in \RR^d$ such that $\ltwonorm{\bv} = 1$, 
			\[
				\begin{aligned}
				|\EE (\widetilde z_i - z_i)(\bv^\top\widetilde\bx_i)| & \le \EE( |z_i|(\bv^\top\bx_i)1_{\{|z_i| > \tau_2\}}) \le \sqrt{\EE (z_i^2(\bv^\top\bx_i)^2)  \PP(|z_i| > \tau_2)}. \\
				& \le (M_1R)^{1/4} \frac{\sqrt{M_1}}{\tau_2^2},
				\end{aligned}
			\]
			thus we have $\ltwonorm{T_{21}} \le M_1^{3/4}R^{1/4}/\tau_2^2$. Again, for any $\bv\in \RR^d$ such that $\ltwonorm{\bv} = 1$, since $\ltwonorm{\EE \epsilon_i \bx_i} \le M_2\sqrt{d/n}$, 
			\[
				\begin{aligned}
					\EE [\epsilon_i (\widetilde\bx_i^\top\bv)] & = \EE [\epsilon_i ((\widetilde\bx_i - \bx_i)^\top\bv)] + \EE [\epsilon_i (\bx_i^\top\bv)] \le \EE [\epsilon_i |\bx_i^\top\bv| 1_{\{\lfournorm{\bx_i} \ge \tau_1\}}] + M_2\biggl ( \frac{d}{n}\biggr ) ^ {1 / 2}\\
					& \le \sqrt{\EE (\epsilon_i(\bx_i^\top\bv))^2 \PP\bigl( \lfournorm{\bx_i}\ge \tau_1 \bigr)} + M_2\biggl ( \frac{d}{n}\biggr ) ^ {1 / 2} \\
					& \le (M_1R)^{1/4} \frac{\sqrt{dR}}{\tau_1^2} + M_2\biggl ( \frac{d}{n}\biggr ) ^ {1 / 2}. 
				\end{aligned}
			\]
			Therefore $\ltwonorm{T_{22}} \le (M_1R)^{1/4}\sqrt{dR}/\tau_1^2 + M_2\sqrt{d/n}$. For $T_{23}$, since $\EE [y_i - b'(\bx_i^\top\bbeta^*)|\bx_i] = 0$, $T_{23} = 0$. Finally we bound $T_{24}$. For any $\bv\in \RR^d$ such that $\ltwonorm{\bv}=1$, 
			\[
				\begin{aligned}
					\ltwonorm{T_{24}} & \le M\E (\bbeta^{*^\top}(\bx_i - \widetilde\bx_i))(\bv^\top\widetilde\bx_i) \le M \E [(\bbeta^{*^\top}\bx_i)(\bv^\top\bx_i)  1_{\{\lfournorm{\bx_i} \ge \tau_1\}}] \\
					& \le M \sqrt{\E (\bbeta^{*^\top} \bx_i)^2(\bv^\top\bx_i)^2 P\bigl(\lfournorm{\bx_i} \ge \tau_1\bigr)} \le ML\sqrt{dR} /\tau_1^2. 
				\end{aligned}
			\]
			To summarize here, we have
			\beq
				\label{eq:5.12}
				\ltwonorm{T_2} \le (M_1R)^{1/4} \bigl(\frac{\sqrt{M_1}}{\tau_2^2} + \frac{\sqrt{dR}}{\tau_1^2}\bigr) + M L\frac{\sqrt{dR}}{\tau_1^2} + M_2\biggl ( \frac{d}{n}\biggr ) ^ {1 / 2}. 
			\eeq
			
			\vspace{.2cm}
			
			\underline{\emph{Bound for $T_3$}}:
				We apply a similar proof strategy as in the bound for $T_1$. Define the following Hermitian dilation matrix: 
			\[
				\widetilde \bX_i := b'(\widetilde\bx_i^\top\bbeta^*) \left( \begin{array}{cc}
					0 & \widetilde\bx_i^\top \\
					\widetilde \bx_i & \bzero
				\end{array}\right). 
			\]
			First, 
			\[
				\opnorm{\EE \widetilde\bX^2_i} = \max(\EE (b'(\widetilde\bx_i^\top\bbeta^*)\widetilde\bx_i^\top\widetilde\bx_i) , \opnorm{\EE b'(\widetilde\bx_i^\top\bbeta^*)^2\widetilde\bx_i\widetilde\bx_i^\top}).
			\]
			Write $|b'(1)|$ as $b_1$. For any $j \in [d]$, 
			\[
				\begin{aligned}
					\EE \bigl(b'(\widetilde\bx_i^\top\bbeta^*)^2  \widetilde x^2_{ij} \bigr) & \le \EE [(b_1 + M|\widetilde\bx_i^\top\bbeta^* - 1|)^2\widetilde x^2_{ij}] \le 2\EE[ ((b_1 + M)^2 + M^2(\widetilde\bx_i^\top\bbeta^*)^2)\widetilde x^2_{ij}] \\
					& \le 2M^2 R\ltwonorm{\bbeta^*}^2 + 2(b_1 + M)^2\sqrt{R} =: V ,
				\end{aligned}
			\]
			so $\EE [b'(\widetilde\bx_i^\top\bbeta^*)^2 \widetilde\bx_i^\top\widetilde\bx_i] \le dV$. In addition, for any $\bv\in \RR^d$ such that $\ltwonorm{\bv} =1$, 
			\[
				\EE \bigl(b'(\widetilde\bx_i^\top\bbeta^*)^2 (\bv^\top\widetilde\bx_i)^2\bigr) \le \EE ((b_1 + M|\widetilde\bx_i^\top\bbeta^* - 1|)^2 (\bv^\top\widetilde\bx_i)^2) \le V. 
			\]
			We thus have $\opnorm{\EE \widetilde \bX_i^2} \le dV$. In addition, $\opnorm{\EE \widetilde \bX_i} = \EE ( b'(\widetilde\bx_i^\top\bbeta^*) \ltwonorm{\widetilde\bx_i}) \le \sqrt{\EE b'(\widetilde\bx_i^\top\bbeta^*)^2 \EE \ltwonorm{\widetilde \bx_i}^2} \le \sqrt{d}V$, which further implies that $\opnorm{\EE (\widetilde \bX_i - \EE \widetilde\bX_i)^{2}} \le (d + \sqrt{d})V$. 
			Also notice that $\opnorm{\widetilde \bX_i} \le ((b_1 + M) + M\ltwonorm{\bbeta^*} d^{1/4}\tau_1)d^{1/4}\tau_1$. By the matrix Bernstein's inequality, 
			\[
				\PP\bigl( \opnorm{\frac{1}{n} \sum\limits_{i=1}^n \widetilde \bX_i - \E \widetilde \bX_i} \ge t \bigr) \le d  \exp\Bigl( -c_1\min\bigl(\frac{nt^2}{(d+ \sqrt{d})V}, \frac{nt}{ (b_1 + M + M\ltwonorm{\bbeta^*} d^{1/4}\tau_1)d^{1/4}\tau_1}\bigr) \Bigr). 
			\]
			
			Given that $\ltwonorm{T_3} = 2\opnorm{n^{-1}\sum\limits_{i=1}^n \widetilde\bX_i - \E \widetilde \bX_i}$, it thus holds that 
			\beq
				\label{eq:5.13}
				\PP\Bigl( \ltwonorm{T_3} \ge 2t \Bigr) \le d \exp\Bigl( -c_1\min\bigl(\frac{nt^2}{(d+ \sqrt{d})V}, \frac{nt}{ (b_1 + M + M\ltwonorm{\bbeta^*} d^{1/4}\tau_1)d^{1/4}\tau_1}\bigr) \Bigr).
			\eeq

			Finally, choose $\tau_1, \tau_2 \asymp (n/ \log n)^{1/4}$. Combining \eqref{eq:5.11}, \eqref{eq:5.12} and \eqref{eq:5.13} delivers that for some constant $C_1$  any $\xi > 1$, 
			\beq
				\label{eq:5.14}
				\PP\biggl \{  \ltwonorm{\nabla \widetilde \ell_n(\bbeta^*)} \ge C_1\xi \biggl ( \frac{d\log n}{n}\biggr ) ^ {1 / 2}\biggr \} \le n ^{1- \xi}. 
			\eeq
		 	Choose $t = \xi \log n$ and let $r$ be larger than the RHS of \eqref{eq:5.9}. When $d / n$ is sufficiently small and $n$ is sufficiently large, we can obtain that 
			\[
				r \ge  C_2\xi  \biggl (  \frac{d \log n}{n}\biggr ) ^ {1 / 2} =:r_0, 
			\]
			where $C_2$ is a constant. Choose $r = r_0$. Then by \eqref{eq:5.9}, $\ltwonorm{\bDelta_{\eta}} \le r_0$ and thus $\widetilde\bDelta = \widetilde\bDelta_{\eta}$. Finally, we reach the conclusion that   
			\[
				\PP\biggl \{  (\ltwonorm{\widetilde\bDelta} \ge  C_2\xi \biggl (  \frac{d \log n}{n}\biggr ) ^ {1 / 2} \biggr \}\le n^{1 - \xi} + 2 n^{-\xi} \le 3n^{1 - \xi} . 
			\]
			
	\end{proof}
	
	\begin{proof}[Proof of Corollary \ref{cor:1}]
		The proof strategy is nearly the same as that for deriving Theorem \ref{thm:1}, so we provide a roadmap here and do not dive into great details. For ease of notation, write $ n^{-1}\sum\limits_{i=1}^n \ell^w(\widetilde\bx_i, z_i; \bbeta)$ as $\widetilde\ell^w(\bbeta)$ and denote the hessian matrix of $\widetilde\ell^w_n(\bbeta)$ by $\widetilde\bH^w_n(\bbeta)$. Since $\widetilde\bH^w_n(\bbeta) = \nabla^2\widetilde\ell_n(\bbeta)= \widetilde\bH_n(\bbeta)$, we can directly obtain the uniform strong convexity of $\widetilde\bH^w_n(\bbeta)$ from Lemma \ref{lem:1}. In addition,
		\[
			\begin{aligned}
				\nabla_{\bbeta} \widetilde \ell^w_n(\bbeta^*) & = \frac{1-p}{1-2p}  \underbrace{\frac{1}{n} \sum\limits_{i=1}^n (b'(\widetilde\bx_i^\top\bbeta^*) - z_i)\widetilde\bx_i}_{T_1} - \frac{p}{1-2p} \underbrace{\frac{1}{n} \sum\limits_{i=1}^n (b'(\widetilde\bx_i^\top\bbeta^*) - (1-z_i))\widetilde\bx_i}_{T_2} \\
				& = \frac{1-p}{1-2p}(T_1 - \EE T_1) - \frac{p}{1-2p}(T_2 - \EE T_2) + \frac{1-p}{1-2p}\EE T_1 - \frac{p}{1-2p}\EE T_2 \\
				& = \frac{1-p}{1-2p}(T_1 - \EE T_1) - \frac{p}{1-2p}(T_2 - \EE T_2) + \EE (b'(\widetilde\bx_i^\top\bbeta^*) - y_i)\widetilde\bx_i. 
			\end{aligned}
		\]
		Since $|b'(\widetilde\bx_i^\top\bbeta^*) - z_i| \le 1$ and $|b'(\widetilde\bx_i^\top\bbeta^*) - (1-z_i)| \le 1$, following the bound for $T_1$ in Theorem \ref{thm:1}, we will obtain 
		\[
			\PP\biggl \{ \ltwonorm{\frac{1-p}{1-2p}(T_1 - \EE T_1) - \frac{p}{1-2p}(T_2 - \EE T_2)} \ge c_1\xi\biggl ( \frac{d \log n}{n}\biggr ) ^ {1 / 2}\biggr \} \le n^{1-\xi},
		\]
		where $c_1>0$ depends on $R$ and $p$ and $\xi > 1$. In addition, following the bound for $T_{23}$ and $T_{24}$ in Theorem 1, we shall obtain
		\[
			\ltwonorm{\EE (b'(\widetilde\bx_i^\top\bbeta^*) - y_i)\widetilde\bx_i} \le M_2L \frac{\sqrt{dR}}{\tau_1^2} \le c_2  M_2\biggl ( \frac{dR\log n}{n}\biggr ) ^ {1 / 2}. 
		\]
		where $c_2>0$ is a constant. Therefore, for some constant $c_3$ depending on $R, p, M_2, R$, we have
		\[
			\PP\biggl \{\ltwonorm{\nabla_{\bbeta} \widetilde\ell_n^w(\bbeta^*)} \ge c_3\xi\biggl (\frac{d\log n}{n}\biggr ) ^ {1 / 2}\biggr \} \le n^{1-\xi}. 
		\]
		Combining this with the uniform strong convexity of $\widetilde \bH^w_n(\bbeta)$ delivers the final conclusion. 
	\end{proof}
	
	\begin{proof}[Proof of Lemma \ref{lem:2}]
		According to \eqref{eq:2.3}, $[\nabla_{\bbeta} \widetilde\ell(\bbeta^*)]_j= (b'(\widetilde\bx_i^\top\bbeta^*)- \widetilde z_i)\widetilde x_{ij}$. Then we have
		\[
			\begin{aligned}
    			\bigl|\frac{1}{n} \sum\limits_{i=1}^n (b'(\widetilde\bx_i^\top\bbeta^*)- \widetilde z_i)\widetilde x_{ij} \bigr| & \le \underbrace{\bigl|\frac{1}{n}\sum\limits_{i=1}^n b'(\widetilde\bx_i^\top\bbeta^*)\widetilde x_{ij}- \EE b'(\widetilde\bx_i^\top\bbeta^*)\widetilde x_{ij}\bigr|}_{T_1} + \underbrace{|\EE (b'(\widetilde\bx_i^\top\bbeta^*)- \widetilde z_i)\widetilde x_{ij}|}_{T_2} \\
			& + \underbrace{\bigl |\frac{1}{n}\sum\limits_{i=1}^n \widetilde z_i\widetilde x_{ij}- \EE \widetilde z_i\widetilde x_{ij}\bigr|}_{T_3} .
			\end{aligned}
		\]
		We start with the upper bound of $T_1$. By the Mean Value Theorem, for any $i \in [n]$, there exists $\xi_i$ between $1$ and $\widetilde x_i^\top\bbeta^*$ such that $b'(\widetilde\bx_i^\top\bbeta^*)= b'(1) + b''(\xi_i)(\widetilde\bx_i^\top\bbeta^* - 1) $. Therefore we have
		\[	
			\begin{aligned}
    			T_1 & \le \bigl| \frac{1}{n} \sum\limits_{i=1}^n b'(1)\widetilde x_{ij}- \EE (b'(1)\widetilde x_{ij}) \bigr| + \bigl|\frac{1}{n} \sum\limits_{i=1}^n b''(\xi_i)\widetilde x_{ij}(\widetilde \bx_i^\top\bbeta^* - 1)- \EE (b''(\xi_i)\widetilde x_{ij}(\widetilde \bx_i^\top\bbeta^* - 1)) \bigr| \\
			& \le  \bigl| \frac{1}{n} \sum\limits_{i=1}^n b'(1)\widetilde x_{ij}- \EE (b'(1)\widetilde x_{ij}) \bigr| + \sum\limits_{k=1}^d |\beta^*_k| \bigl |\frac{1}{n} \sum\limits_{i=1}^n b''(\xi_i)\widetilde x_{ij}\widetilde x_{ik}- \EE b''(\xi_i)\widetilde x_{ij}\widetilde x_{ik} \bigr| \\
			& + \bigl | \frac{1}{n}\sum\limits_{i=1}^n b''(\xi_i)\widetilde x_{ij} - \EE( b''(\xi_i)\widetilde x_{ij} ) \bigr |. 
			\end{aligned}
		\]
		Since $\var (\widetilde x_{ij}) \le  \sqrt{R}$ and $|\widetilde x_{ij}| \le \tau_1$, an application of Bernstein's inequality (Theorem 2.10 in \cite{BLM13}) yields that
		\[
    			\PP\biggl[ |\frac{1}{n}\sum\limits_{i=1}^n b'(1)\widetilde x_{ij}- \EE(b'(1)\widetilde x_{ij}) |\ge |b'(1)|\biggl\{ \biggl ( \frac{\sqrt{R}  2t}{n}\biggr ) ^ {1 / 2}+\frac{c_1\tau_1 t}{n}\biggr \} \biggr]\le 2\exp(-t),
		\]
		where $c_1>0$ is some universal constant. In addition, $b''(\xi_i)\widetilde x_{ij}\widetilde x_{ik} \le M \tau^2_1$ and $\var (b''(\xi_i)\widetilde x_{ij}\widetilde x_{ik}) \le \E (b''(\xi_i)\widetilde x_{ij}\widetilde x_{ik})^2 \le M^2R$. Again by Bernstein's inequality, 
		\[
			\PP\biggl\{  \bigl |\frac{1}{n} \sum\limits_{i=1}^n b''(\xi_i)\widetilde x_{ij}\widetilde x_{ik}- \EE (b''(\xi_i)\widetilde x_{ij}\widetilde x_{ik}) \bigr| \ge \biggl ( \frac{2M^2 R t}{n}\biggr ) ^ {1 / 2}+ \frac{c_1M\tau^2_1 t}{n}\biggr \} \le 2\exp(-t). 
		\]
		Similarly, 
		\[	
			\PP\biggl \{  \bigl |\frac{1}{n} \sum\limits_{i=1}^n b''(\xi_i)\widetilde x_{ij} - \EE (b''(\xi_i)\widetilde x_{ij}) \bigr | \ge \biggl ( \frac{M^2\sqrt{R}t}{n}\biggr ) ^ {1 / 2} + \frac{M\tau_1 t}{n}\biggr \} \le 2\exp(-t). 
		\]
		Combining the above three inequalities delivers that 
		\beq
			\label{eq:5.15}
			\begin{aligned}
				\PP \biggl[ T_1 \ge |b'(1)| \biggl\{\biggl ( \frac{\sqrt{R} 2t}{n}\biggr ) ^ {1 / 2}+ \frac{c_1\tau_1 t}{n}\biggr\}  + \biggl ( \frac{2M^2 R t}{n}\biggr ) ^ {1 / 2} + \frac{c_1M\tau^2_1 t}{n} + \biggl ( \frac{M^2\sqrt{R}t}{n}\biggr ) ^ {1 / 2} & + \frac{M\tau_1 t}{n}  \biggr] \\
				&\le 6\exp(-t). 
			\end{aligned}
		\eeq
		Now we bound $T_2$. 
		\beq
			\label{eq:5.16}
			\begin{aligned}
				T_2 & = \EE [(z_i - \widetilde z_i)\widetilde x_{ij}] + \EE \epsilon_i \widetilde x_{ij} + \EE [(b'(\bx_i^\top\bbeta^*) - b'(\widetilde\bx_i^\top\bbeta^*))\widetilde x_{ij}]\\
			& \le \EE [|z_i \widetilde x_{ij}| 1_{\{ |z_i| \ge \tau_2\}}] + \EE (\epsilon_i x_{ij} ) + \EE \epsilon_i(x_{ij} - \widetilde x_{ij}) + M \sum\limits_{k=1}^d |\beta^*_k| \EE |\widetilde x_{ik}(\widetilde x_{ij}- x_{ij})|  \\
			& \le (M_1R)^{1/4} \frac{\sqrt{M_1}}{\tau_2^2} + \frac{M_3}{\sqrt{n}} + \frac{(M_1 R)^{1/4}}{\tau_1^2} + M M_2 \frac{\sqrt{R}}{\tau_1^2}. 
			\end{aligned}
		\eeq
		Finally we bound $T_3$. Note that $|\widetilde z_i \widetilde x_{ij}| \le \tau_1\tau_2$, $\var(\widetilde x_{ij} \widetilde z_i) \le \E |\widetilde z_i\widetilde x_{ij}|^2 \le \sqrt{M_1R}$. According to the Bernstein's inequality, 
		\beq
			\label{eq:5.17}
			\PP\biggl\{  |T_3| \ge \biggl ( \frac{2t\sqrt{M_1 R}}{n}\biggr ) ^ {1 / 2} + \frac{c_1\tau_1\tau_2 t}{n} \biggr \} \le 2\exp(-t). 
		\eeq
		Choose $\tau_1, \tau_2 \asymp (n /\log d)^{1/4}$.  Combining \eqref{eq:5.15}, \eqref{eq:5.16} and \eqref{eq:5.17} delivers that for some constant $C_1>0$ that depends on $M, R, \{M_i\}_{i=1}^3, b'(1)$ and any $\xi>1$, 
		\[
			\PP\biggl \{  |[\nabla_{\bbeta} \widetilde \ell(\bbeta^*)]_j| \ge C_1\xi \biggl ( \frac{\log d}{n}\biggr ) ^ {1 / 2}\biggr \} \le 2d^{-\xi}. 
		\]
		Then by the union bound for all $j \in [d]$, it holds that
		\[
			\PP\biggl \{\max_{j\in [d]} [|\nabla_{\bbeta} \widetilde \ell(\bbeta^*)]_j| \ge C_1\xi\biggl (\frac{\log d}{n}\biggr ) ^ {1 / 2}\biggr \} \le 2 d^{1-\xi}. 
		\]
		\end{proof}
		
		\begin{proof}[Proof of Lemma \ref{lem:3}]
		
		The proof strategy is quite similar to that for Lemma \ref{lem:1}, except that we need to take advantage of the restricted cone $\cC(\cS)$ that $\bDelta$ lies in. First of all, for any $1 \le j, k \le d$, 
		\[
			|\EE (\widetilde x_{ij}\widetilde x_{ik} - x_{ij}x_{ik})| \le \sqrt{\EE (x_{ij}x_{ik})^2  (\PP(|x_{ij}| \ge \tau_1) + \PP(|x_{ik}| \ge \tau_1))} \le \frac{\sqrt{2} R}{\tau^2_1}. 
		\] 
		We thus have
		\beq
			\label{eq:5.18}
			\|\EE [\bx_i\bx_i^\top - \widetilde\bx_i\widetilde\bx_i^\top]\|_{\max} \le \frac{\sqrt{2}R}{\tau_1^2} \le CR\biggl ( \frac{2\log d}{n}\biggr ) ^ {1 / 2},
		\eeq
		where $C > 0$ is some constant. Again, define a contraction function 
		\[
			\phi(x; \theta) = x^2  \ind_{\{ |x| \le \theta\}} + (x - 2\theta)^2  \ind_{\{ \theta < x \le 2\theta\}} + (x + 2\theta)^2  \ind_{\{ -2\theta \le x < -\theta\}}. 
		\]
		 Given any $\bDelta \in \cB_2(\bzero, r) \cap \cC(\cS)$, by the Taylor expansion, we can find $v \in (0, 1)$ such that
		\beq
			\label{eq:5.19}
			\begin{aligned}
				& \delta\widetilde \ell_n(\bbeta^* + \bDelta; \bbeta^*)  = \widetilde \ell_n(\bbeta^* + \bDelta) - \widetilde \ell_n(\bbeta^*) - \nabla\widetilde\ell_n(\bbeta^*)^\top\bDelta = \frac{1}{2} \bDelta^\top \widetilde\bH_n(\bbeta^* + v\bDelta)\bDelta \\
				& = \frac{1}{2n}\sum\limits_{i=1}^n b''(\widetilde \bx_i^\top (\bbeta^* + v\bDelta)) (\bDelta^\top \widetilde\bx_i)^2 \ge \frac{1}{2n}\sum\limits_{i=1}^n b''(\widetilde \bx_i^\top (\bbeta^* + v\bDelta)) \phi(\bDelta^\top \widetilde\bx_i; \alpha_1 r) \ind_{\{ |{\bbeta^*}^\top \widetilde\bx_i| \le \alpha_2 \} } \\
				& \ge \frac{m(\omega)}{2n}\sum\limits_{i=1}^n \phi(\bDelta^\top \widetilde\bx_i; \alpha_1 r) \ind_{\{ |{\bbeta^*}^\top \widetilde\bx_i| \le \alpha_2 \} }, 
			\end{aligned}
		\eeq
		where we choose $\omega = \alpha_1 + \alpha_2 > \alpha_1 r + \alpha_2$ so that the last inequality holds by Condition (1). For ease of notation, let $\cA_i := \{|\bDelta^\top\widetilde \bx_i| \le \alpha_1 r\}$ and $\cB_i := \{|{\bbeta^*}^\top\widetilde \bx_i| \le \alpha_2\}$. We have
		\[
			\begin{aligned}
				\EE & [\phi(\bDelta^\top \widetilde\bx_i; \alpha_1 r) \ind_{\cB_i}] \ge \EE [(\bDelta^\top \widetilde \bx_i)^2  \ind_{\cA_i \cap \cB_i}] \\
			 & \ge \bDelta^\top \EE [\bx_i\bx_i^\top \ind_{\cA_i \cap \cB_i}] \bDelta  - \bDelta^\top\EE [(\bx_i\bx_i^\top - \widetilde\bx_i\widetilde\bx_i^\top)  \ind_{\cA_i \cap \cB_i}] \bDelta \\
			& \ge \bDelta^\top \EE [\bx_i\bx_i^\top \ind_{\cA_i \cap \cB_i}] \bDelta  - \bDelta^\top\EE [\bx_i\bx_i^\top - \widetilde\bx_i\widetilde\bx_i^\top] \bDelta \\
			& \ge \bDelta^\top \EE (\bx_i\bx_i^\top)\bDelta - \bDelta^\top\EE (\bx_i\bx_i^\top \ind_{\cA^c_i \cup \cB^c_i})\bDelta - \bDelta^\top\EE [\bx_i\bx_i^\top - \widetilde\bx_i\widetilde\bx_i^\top] \bDelta \\
			& \ge \kappa_0 \ltwonorm{\bDelta}^2 - \sqrt{\EE (\bDelta^\top \bx_i) ^4  (\PP(\cA_i^c) + \PP(\cB_i^c))} - \bDelta^\top\EE [\bx_i\bx_i^\top - \widetilde\bx_i\widetilde\bx_i^\top] \bDelta \\
			& \ge \kappa_0 \ltwonorm{\bDelta}^2 - \sqrt{R(\PP(\cA_i^c) + \PP(\cB_i^c))}\ltwonorm{\bDelta}^2 -  \|\EE [\bx_i\bx_i^\top - \widetilde\bx_i\widetilde\bx_i^\top]\|_{\max} \lonenorm{\bDelta}^2
			\end{aligned}
		\]
		By the Markov Inequality and \eqref{eq:5.18}, 
		\[
			\begin{aligned}
			\PP(\cA^c_i) & \le \frac{\EE (\bDelta^\top \widetilde \bx_i)^2}{\alpha^2_1 r^2} \le \frac{\EE (\bDelta^\top\bx_i)^2 + \bDelta^\top \EE(\widetilde \bx_i\widetilde \bx_i^\top - \bx_i\bx_i^\top)\bDelta }{\alpha^2_1r^2}  \\
			& \le \frac{\sqrt{R}\ltwonorm{\bDelta}^2 + CRs\ltwonorm{\bDelta}^2 \sqrt{{2\log d}/{n}} }{\alpha_1^2r^2} \le \frac{\sqrt{R} + CRs\sqrt{\log d / n}}{\alpha_1^2}
			\end{aligned}
		\]
		and
		\[
			\begin{aligned}
				\PP(\cB^c_i) & \le \frac{\EE ({\bbeta^*}^\top \widetilde \bx_i)^2}{\alpha_2^2} \le \frac{\EE ({\bbeta^*}^\top\bx_i)^2 + {\bbeta^*}^\top \EE(\widetilde \bx_i\widetilde \bx_i^\top - \bx_i\bx_i^\top)\bbeta^* }{\alpha^2_2}  \\
			& \le \frac{\sqrt{R}\ltwonorm{\bbeta^*}^2 + CRs\ltwonorm{\bbeta^*}^2 \sqrt{{2\log d}/{n}} }{\alpha_2^2} \le \frac{\sqrt{R}L^2 + CRL^2s\sqrt{2 \log d / n}}{\alpha_2^2}. 
			\end{aligned}
		\]
		Overall, as long as $\alpha_1, \alpha_2$ are sufficiently large and $s\sqrt{\log d/n}$ is not large, 
		\beq
			\label{eq:5.20}
			\EE [\phi(\bDelta^\top \widetilde\bx_i; \alpha_1 r)  \ind_{\cB_i}] \ge \frac{\kappa_0}{2} \ltwonorm{\bDelta}^2. 
		\eeq
		For notational convenience, define $Z_i := \phi(\bDelta^\top \widetilde\bx_i; \alpha_1 r)  \ind_{ \cB_i } = \phi(\bDelta^\top \widetilde \bx_i   \ind_{ \cB_i }; \alpha_1 r)$ and $\Gamma_r := \sup_{\ltwonorm{\bDelta} \le r, \bDelta \in \cC(\cS)} \bigl | n^{-1} \sum\limits_{i=1}^n Z_i - \EE Z_i \bigr|$. Then an application of Massart's inequality (\cite{Mas00}) delivers that
		\beq
			\label{eq:5.21}
			\PP\biggl \{ |\Gamma_r - \EE \Gamma_r | \ge \alpha^2_1r^2\biggl ( \frac{t}{n} \biggr ) ^ {1 / 2}\biggr \} \le 2\exp\bigl( - \frac{t}{8}\bigr). 
		\eeq
		The remaining job is to derive the order of $\EE \Gamma_r$. By the symmetrization argument and Ledoux-Talagrand contraction inequality,  for a sequence of i.i.d. Rademacher variables $\{\gamma_i\}_{i=1}^n$, 
		\[
			\begin{aligned}
			\EE \Gamma_r & \le 2\EE \sup_{\ltwonorm{\bDelta} \le r, \bDelta \in \cC(\cS) } \bigl | \frac{1}{n} \sum\limits_{i=1}^n \gamma_iZ_i \bigr | \le 8\alpha_1 r  \EE \sup_{\ltwonorm{\bDelta} \le r, \bDelta \in \cC(\cS)} \bigl| \inn{\frac{1}{n}\sum\limits_{i=1}^n \gamma_i \widetilde \bx_i  \ind_{\{|{\bbeta^*}^\top \widetilde\bx_i| \le \alpha_2\}}, \bDelta}  \bigr| \\
			& \le 8\alpha_1\sqrt{s} r^2  \EE \supnorm{ \frac{1}{n}\sum\limits_{i=1}^n \gamma_i \widetilde \bx_i   \ind_{\{|{\bbeta^*}^\top \widetilde\bx_i| \le \alpha_2\}}}. 
			\end{aligned}
		\]
		For any $1 \le j \le d$, by Bernstein inequality, 
		\[
			\PP\biggl \{  |\frac{1}{n} \sum\limits_{i=1}^n \gamma_i\widetilde x_{ij}  \ind_{\{|{\bbeta^*}^\top \widetilde\bx_i| \le \alpha_2\}} | \ge \biggl ( \frac{2\sqrt{R}t}{n}\biggr ) ^ {1 / 2} + \frac{C_1\tau_1 t}{n}\biggr \} \le 2\exp(-t),
		\]
		where $C_1$ is some constant. By the union bound, we can deduce that for some constant $C_2$, 
		\[
			\PP\biggl\{  \supnorm{\frac{1}{n} \sum\limits_{i=1}^n \gamma_i \widetilde \bx_i  \ind_{\{|{\bbeta^*}^\top \widetilde\bx_i| \le \alpha_2\}}} \ge C_2\biggl ( \frac{t\log d}{n} \biggr ) ^ {1 / 2}\biggr \} \le 2d^{1-t},
		\]
		which further implies that 
		\[
			\EE \Gamma_r \le 8 \alpha_1 \sqrt{s} r^2  \EE \supnorm{\frac{1}{n} \sum\limits_{i=1}^n \gamma_i \widetilde \bx_i  \ind_{\{|{\bbeta^*}^\top \widetilde\bx_i| \le \alpha_2\}}} \le 8C_3\alpha_1r^2 \biggl ( \frac{s\log d}{n}\biggr ) ^ {1 / 2}. 
		\]
		for some constant $C_3$. Combining the above inequality with \eqref{eq:5.19}, \eqref{eq:5.20} and \eqref{eq:5.21} yields that for any $t>0$, with probability at least $1 - 2\exp(-t)$, 
		\[
			\delta \widetilde\ell_n(\bbeta; \bbeta^*) \ge \frac{m\kappa_0}{4}\ltwonorm{\bDelta}^2 - \alpha_1^2  r^2 \biggl ( \frac{8t}{n}\biggr ) ^ {1 / 2} - 8C_3\alpha_1r^2\biggl ( \frac{s\log d}{n}\biggr ) ^ {1 / 2}. 
		\]

		\end{proof}
		
		\begin{proof}[Proof of Theorem \ref{thm:2}]
		
		According to Lemma 1 in \cite{NRW12}, as long as $\lambda \ge 2\supnorm{\nabla \widetilde \ell_n(\bbeta)}$,  $\widetilde \bDelta \in \cC(\cS)$. We construct an intermediate estimator $\widetilde\bbeta_{\eta} $ between $\widetilde\bbeta$ and $\bbeta^*$:
		\[
			\widetilde \bbeta_{\eta} = \bbeta^* + \eta (\widetilde \bbeta - \bbeta^*), 
		\]
		where $\eta = 1$ if $\ltwonorm{\widetilde \bbeta - \bbeta^*} \le r$ and $\eta  = r/\ltwonorm{\widetilde\bbeta - \bbeta^*}$ if $\ltwonorm{\widetilde\bbeta - \bbeta^*} > r$. Choose $\lambda = 2 C\xi\sqrt{\log d / n}$, where $C$ and $\xi$ are the same as in Lemma \ref{lem:2}. By Lemmas \ref{lem:2} and \ref{lem:3}, it holds with probability at least $1 - 2\exp(-t)$ that 
	\beq
		\begin{aligned}
			\kappa \ltwonorm{\widetilde\bDelta_{\eta}}^2 - C_0 r^2\biggl \{ \biggl ( \frac{t}{n}\biggr ) ^ {1 / 2} + \biggl ( \frac{s \log d}{n} \biggr ) ^ {1 / 2}\biggr \} & \le \delta\widetilde\ell_n(\widetilde\bbeta_{\eta}; \bbeta^*) \\ & = \widetilde\ell_n(\widetilde\bbeta_{\eta}) - \widetilde\ell_n( \bbeta^*) - \nabla\widetilde\ell_n(\bbeta^*)^\top\widetilde\bDelta_{\eta} \\ & \le \lambda \lonenorm{\widetilde \bDelta_{\eta}} + \supnorm{\nabla \widetilde \ell_n(\bbeta^*)}  \lonenorm{\widetilde \bDelta_{\eta}}\\ & \le (\lambda  + \supnorm{\nabla \widetilde \ell_n(\bbeta^*)})  \lonenorm{\widetilde \bDelta_{\eta}} 
			\\ & \le 4 (\lambda  + \supnorm{\nabla \widetilde \ell_n(\bbeta^*)}) \lonenorm{[\widetilde \bDelta_{\eta}]_{\cS}} 
			\\ &\le 4\sqrt{s} (\lambda  + \supnorm{\nabla \widetilde \ell_n(\bbeta^*)})  \ltwonorm{\widetilde \bDelta_{\eta}}. 
		\end{aligned}
	\eeq
	Some algebra delivers that
	\beq
		\label{eq:5.24}
		\begin{aligned}
			\ltwonorm{\widetilde \bDelta_{\eta}} & \le  \frac{4\sqrt{s} (\lambda  + \supnorm{\nabla \widetilde \ell_n(\bbeta^*)})}{\kappa} + r\biggl [\frac{C_0}{\kappa} \biggl\{  \biggl ( \frac{t}{n}\biggr ) ^ {1 / 2} + \biggl ( \frac{s \log d}{n}\biggr ) ^ {1 / 2} \biggr \}\biggr ] ^ {1 / 2} \\ & =
			\frac{4\sqrt{s}  \supnorm{\nabla \widetilde \ell_n(\bbeta^*)}}{\kappa} +\frac{8C\xi}{\kappa}\biggl ( \frac{s \log d}{n}\biggr ) ^ {1 / 2} + r\biggl [\frac{C_0}{\kappa} \biggl\{  \biggl ( \frac{t}{n}\biggr ) ^ {1 / 2} + \biggl ( \frac{s \log d}{n}\biggr ) ^ {1 / 2} \biggr \}\biggr ] ^ {1 / 2}. 
		\end{aligned} 
	\eeq
	Choose $t = \xi \log d$ above. Let $r$ be greater than the RHS of the inequality above. For sufficiently  sufficiently small $s\log d/n$, we have $r \ge 5\sqrt{s} \supnorm{\nabla \widetilde \ell_n(\bbeta^*)} / \kappa$. Define $r_0 := 5\sqrt{s} \supnorm{\nabla\widetilde \ell_n(\bbeta^*)} / \kappa$ and choose $r = r_0$. Therefore, $\ltwonorm{\widetilde \bDelta_{\eta}} \le r$ and $\widetilde \bDelta_{\eta} = \widetilde \bDelta$. By Lemma \ref{lem:2}, we reach the conclusion. 
	\end{proof}

\end{document}